\documentclass[letterpaper,twocolumn,10pt]{article}
\usepackage{usenix}

\usepackage{textcomp}
\usepackage{setspace}
\usepackage{latexsym,fancyhdr,url}
\usepackage{enumerate}
\usepackage[linesnumbered,ruled]{algorithm2e}
\usepackage{graphics}
\usepackage{xparse} %
\usepackage{xspace}
\usepackage{multirow}
\usepackage{csvsimple}
\usepackage{balance}

\usepackage{amsmath}
\usepackage{booktabs}
\usepackage{enumitem}
\usepackage{colortbl}
\usepackage{amsmath,amssymb,amsfonts}
\usepackage{fancyhdr}
\usepackage{makecell}
\usepackage{xcolor, eucal}
\usepackage{array}
\usepackage{subfigure}
\usepackage{graphicx}
\usepackage{caption}
\usepackage{tabularx}
\usepackage{dsfont}
\definecolor{gray0}{gray}{0.9}
\definecolor{mine}{RGB}{205, 232, 248} 
\usepackage{algpseudocode}
\usepackage{colortbl}

\usepackage{amsmath,amssymb,amsfonts}
\usepackage{graphicx}
\usepackage{textcomp}
\usepackage{booktabs}
\usepackage{hyperref}
\usepackage{multirow}
\usepackage{makecell}
\usepackage{url}

\usepackage{amsmath,amsthm}
\usepackage{epsfig,endnotes}
\usepackage{xspace}
\usepackage{enumitem}
\usepackage{cleveref}
\usepackage{float}

\definecolor{gray0}{gray}{0.9}
\definecolor{mine}{RGB}{205, 232, 248} 
\def \toolname{DP-SAPF}

\newcommand\gc[1]{{\color{black}  #1}}

\newtheorem{definition}{\bf Definition}[]
\newtheorem{theorem}{\bf Theorem}

\usepackage{filecontents}

\begin{document}

\date{}

\title{\Large \bf DP-SAPF: Saliency-Aware Parameter Fine-tuning of Public Models for Differentially Private Image Synthesis}

\author{
{\rm Chen Gong}\\
University of Virginia
\and
{\rm Kecen Li}\\
NUS
\and
{\rm Zinan Lin}\\
Microsoft Research
\and
{\rm Tianhao Wang}\\
University of Virginia
} %

\vspace{-5mm}
\maketitle

\begin{abstract}
Differentially private (DP) image synthesis 
generates images that preserve the statistical characteristics of a sensitive dataset, enabling sensitive data analysis and usage while providing rigorous guarantees of privacy leakage. Existing methods fine-tune public models using DP Stochastic Gradient Descent (DP-SGD) on sensitive images to generate synthetic images. But full fine-tuning public models on sensitive images is computationally expensive, because current public models typically contain a large number of parameters. Recent work proposes heuristically using Low-Rank Adaptation (LoRA) on all attention-layer parameters of public models to reduce the number of trainable parameters. However, we argue that exhaustive LoRA coverage across all attention-layer parameters is suboptimal in a DP setting, as it leads to noise accumulation and collapse during private training. 

To address this issue, we propose \toolname, which uses a saliency-aware strategy to identify specific target parameters for LoRA training under DP. \toolname{} is inspired by the fact that larger gradients signify higher saliency, indicating that these parameters are most critical for the DP learning. Specifically, we feed the sensitive images into public models, compute gradients, and add noise to the gradients to satisfy DP. Then, \toolname{} identifies the most salient parameters, those exhibiting high gradient magnitudes on sensitive images, for DP fine-tuning. Experiments on four sensitive image datasets show that \toolname{} improves the utility and fidelity of synthetic images while requiring fewer computational resources than fine-tuning methods without parameter selection. 

\end{abstract}

\section{Introduction}
\label{sec:intro}

Differentially private (DP) image synthesis generates artificial images that maintain the statistical characteristics of sensitive datasets under DP~\cite{gong2025dpimagebench}. This allows for dataset sharing while providing rigorous guarantees to reduce privacy leakage of sensitive images. The assistance of public resources has catalyzed the rapid development of DP image synthesis~\cite{dplora,dpldm,gong2025dpimagebench,pe3,dpsda,gong2025privorl}.
One direction is pretraining synthesizers on public datasets and then finetuning models on sensitive images with Differentially Private Stochastic Gradient Descent (DP-SGD)~\cite{dp-diffusion,dpldm,li2023privimage}.
However, prior work~\cite{gong2025dpimagebench} shows that pretraining synthesizers on public datasets is time-consuming and often fails to synthesize high-resolution images with high quality. For example, DPImageBench~\cite{gong2025dpimagebench} reports that PrivImage (a predominant DP image synthesis method using public datasets)~\cite{li2023privimage} achieves only 61.2\% downstream classification accuracy on {\tt CelebA}~\cite{celeba} datasets with 
$64 \times 64$ resolution, and the training of the synthesizer is time-consuming.

To tackle the challenges, the other direction is to directly use \emph{public models} for DP fine-tuning~\cite{dplora}. %
The public models like Stable Diffusion~\cite{labelembedding} are pretrained on massive datasets using vast computational resources, allowing them to capture robust generative priors, represent complex visual distributions, and generate high-resolution images. These public models are readily accessible via open-source platforms like HuggingFace. %
However, the parameter-heavy nature of public models makes fine-tuning on sensitive data computationally expensive. For example, Stable Diffusion-v1-5\footnote{\url{https://huggingface.co/stable-diffusion-v1-5} \label{sdv15}} (a public model) has about 1 billion parameters~\cite{labelembedding}.

To address this issue, existing DP image synthesis methods using public models~\cite{dplora,dpldm,gong2025dpimagebench,pe3,dpsda,liu2025privcode} contribute two paradigms based on their reliance on fine-tuning. The first category comprises fine-tuning-free methods~\cite{dpsda,xiedifferentially,wang2025synthesize,wang2025struct,pe3}, like PE~\cite{dpsda}. This approach iteratively guides public models to generate synthetic images that align with sensitive data by selecting and refining the most similar synthetic images. However, these fine-tuning-free methods suffer from quality degradation when the target synthetic distribution significantly diverges from the sensitive images~\cite{gong2025dpimagebench,pe3}. A second category involves fine-tuning-based methods~\cite{dplora,dpldm,liu2025privcode,zou2026pe}, like DP-LoRA~\cite{dplora}. These methods mitigate high computational costs by using LoRA~\cite{lora} under DP to reduce the number of trainable parameters in public models. Existing DP synthesis methods neither perform LoRA on all parameters of attention layers
fine-tuning~\cite{dplora,labelembedding} nor following established practices in non-DP image synthesis~\cite{peft,nondpfinetuning}.
They heuristically optimize only the attention layers~\cite{dpldm,dplora}.

\gc{However, this paper claims that exhaustive fine-tuning coverage across all attention layers is suboptimal in a DP setting, as it leads to excessive noise accumulation and training collapses.} For example, as shown in Section~\ref{subsec:mot}, the attention layers within the middle block of the Stable Diffusion exhibit high empirical sensitivity. These layers act as refining encoder-extracted features and ensuring semantic alignment with text, where even parameter updates under DP may significantly steer the synthetic performance~\cite{labelembedding}. \gc{Thus, we investigate: \textit{how to relieve the training collapse when using public models in DP fine-tuning to optimize the synthetic performance?} This challenge has not been explicitly addressed in prior DP image synthesis methods~\cite{gong2025dpimagebench,dplora,dpldm,pe3,li2023privimage}.}

To answer this question, this paper proposes Saliency-Aware Parameter Fine-tuning of public models for DP image synthesis (\toolname), which 
addresses two fundamental challenges in private fine-tuning.

\begin{itemize}[leftmargin=*]
\item \textit{Regarding the selection mechanism}. 
When finetuning public models on sensitive images, some parameters already equipped with general features exhibit very small gradients, indicating that little adjustment is needed, whereas parameters suffering from domain mismatch show increased gradients~\cite{wang2020picking}. In such small‑gradient regions, the injected DP noise can easily dominate the true gradient and derail training. 
\toolname{} uses these gradients as an analytical proxy. 
By selecting parameters with large gradients (with DP), it can bypass fragile layers and maintain training stability.

\item \textit{Concerning the structural granularity of selection}. Moreover, \toolname{} transitions from coarse layer-wise tuning to a matrix-wise selection strategy. This design choice exploits the functional heterogeneity within attention modules, where different weight matrices (even within a layer, there are different parts, modeled by matrices) 
exert disproportionate influence on fine-tuning performance~\cite{attentionlayerimproces}. Thus, matrix-level selection enables a more fine-grained balancing between effective fine-tuning and the perturbations introduced by DP noise.
\end{itemize}

\Cref{sec:ass} presents more motivation and technical details of \toolname{}. \Cref{fig:sapf} presents the framework of \toolname.

We conduct extensive experiments across four sensitive datasets and four public models, presenting that \toolname{} improves downstream classification accuracy and reduces the FID score, relative to baselines. For instance, under $\epsilon=10.0$ and using `Stable-Diffusion-v1-5\textsuperscript{\ref{sdv15}}' as the public model, the FID score is reduced by 83.1\%, and the accuracy increases by 31.5\% compared to DP-LoRA averaged across four sensitive image datasets. We conduct analysis for various privacy budget shows that the synthetic performance is insensitive to the selection ratio within the range of [20\%, 50\%]. Thus, tuning the selection ratio in practice is straightforward. We observe that the parameter matrices closer to the model’s input and output layers exhibit larger gradients, suggesting that these matrices play more critical roles in private training.
\begin{itemize}[leftmargin=*]
    \item This paper proposes \toolname{}, the first framework to introduce fine-grained dynamically parameter selection into DP fine-tuning. We present the observation that empirical parameter selection from prior works can render DP fine-tuning unstable. \toolname{} paves the way for future research in DP image synthesis.
    \item \toolname{} introduces a saliency-aware parameter selection method for DP fine-tuning. By dynamically identifying the most critical parameters based on the gradients derived from different sensitive image datasets under DP, \toolname{} reduces the number of parameters required for fine-tuning, facilitating efficient and stable DP image synthesis.
    \item Extensive experiments across four sensitive datasets present that \toolname{} improves utility and fidelity of synthetic images, compared to fine-tuning methods without parameter selection, while consuming fewer computational resources.
\end{itemize}

\section{Background}
\label{sec:setup}

This section introduces the background, including the DP notion, diffusion models, current DP image synthesis methods that take advantage of public models, and the LoRA method.

\subsection{Differential Privacy}
\label{sub:dp}
The Differential Privacy (DP)~\cite{dp} provides a robust mathematical framework for quantifying the privacy risk associated with data leakage. The concept of DP is defined as follows.

\begin{definition}[\textit{Differential Privacy}~\cite{dp}]
     A randomized algorithm $\mathcal{Q}$ achieves $(\varepsilon, \delta)$-DP if, for any two neighboring datasets $D$ and $D'$ (one can be derived from the other by \textit{adding or removing} a single record), and for all possible output sets $\mathcal{O} \subseteq Range(\mathcal{Q})$, the following inequality is satisfied,
\begin{equation}\label{eq:dp}
    \Pr[\mathcal{Q}(D) \in \mathcal{O}] \leq e^\varepsilon \Pr[\mathcal{Q}(D') \in \mathcal{O}] + \delta.
\end{equation}
\end{definition}
\noindent Here, $\varepsilon > 0$ denotes the privacy budget, where lower values correspond to more stringent privacy protections. The parameter $\delta \geq 0$ means the probability of a privacy breach beyond the $\varepsilon$ guarantee. A fundamental property of DP is its immunity to post-processing~\cite{dp}: for any data-independent function $\mathcal{F}$, the composition $\mathcal{F} \circ \mathcal{Q}$ maintains the same $(\epsilon, \delta)$-DP guarantee. This ensures that downstream analysis performed on synthetic images satisfying DP mechanisms, incurs no additional privacy degradation.

\vspace{1mm}
\noindent \textbf{DP-SGD.} We adopt DP-SGD~\cite{dpsgd} to satisfy privacy constraints while fine-tuning our DP image synthesizer. This mechanism sanitizes the learning process by first bounding the $\ell_2$ norm of individual gradients to $C$ and then perturbing the aggregated result with Gaussian noise $\mathcal{N}(0, \sigma_d^2 \mathbb{I})$. At each step, we utilize a Poisson sampling strategy (sampling rate $q$) to extract a sub-batch of images $D_s^{\text{sub}}=\{x_i\}_{i=1}^{B}$ from the sensitive dataset $D_s$. The realized batch size $B$ is a random variable following a Binomial distribution, whereas $B^* = qN$ denotes the expected batch size. Following the standard practice in DP-SGD, $B^*$ is the fixed normalization constant for the gradient sum to ensure the unbiasedness of the estimator and consistency in privacy accounting~\cite{dpsgd}. The parameters $\theta$ of the synthesizer are updated via the following noisy gradient, 
\begin{equation}
\label{eq:dpsgd}
    \lambda \left( \frac{1}{{B^*}} \sum_{i = 1}^{{B}} \text{Clip}\left(\nabla {\mathcal{L}}(\theta, x_i), C\right) + \frac{C}{{B^*}} \mathcal{N}(0, \sigma_d^2 \mathbb{I}) \right),
\end{equation}
\noindent where $\mathcal{L}$ is the objective function of DP image synthesizer, and $\lambda$ is the learning rate and $\sigma_d^2$ is the variance of Gaussian noise. The $\text{Clip}(\cdot, C)$ denotes the norm-clipping operation that scales the gradient $\nabla \mathcal{L}$ if its $\ell_2$ norm exceeds the threshold $C$. Following previous works~\cite{gong2025dpimagebench,dpldm,dplora,dp-diffusion}, we use R\'{e}nyi DP (RDP)~\cite{sgm} to track the privacy loss, as detailed in~\Cref{app:supp_dp}.

\begin{figure}[!t]
    \centering
    \setlength{\abovecaptionskip}{0pt}
    \includegraphics[width=0.97\linewidth]{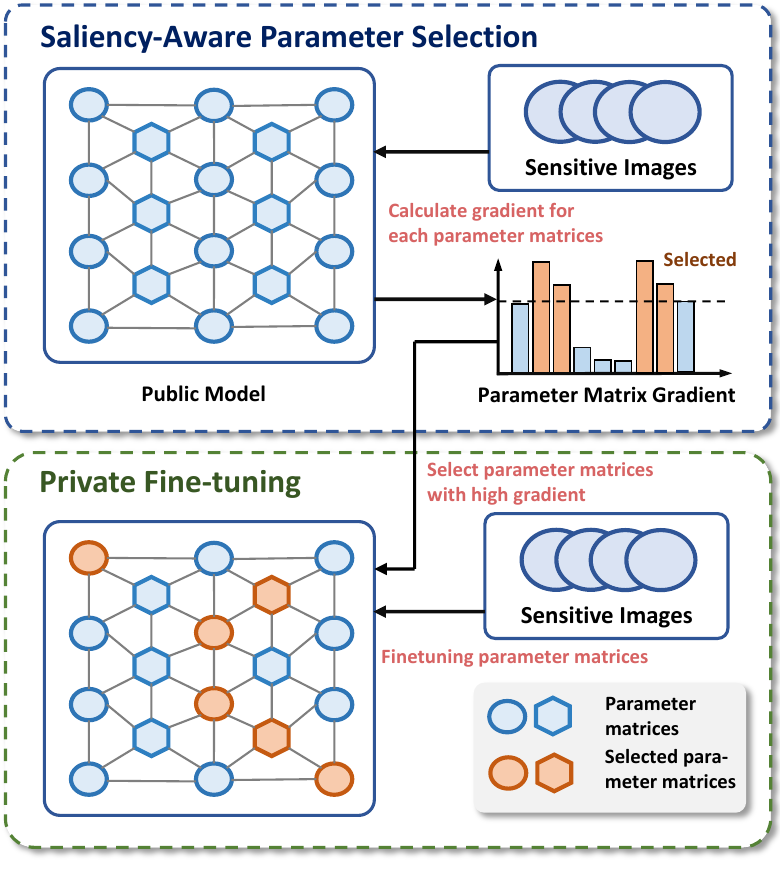}
    \vspace{1mm}
    \caption{The framework of \toolname. In the first stage, \toolname{} computes the gradient of sensitive images with respect to each parameter matrix. Parameter matrices exhibiting large gradients are identified as salient and selected for fine‑tuning. In the second stage, these selected parameter matrices are updated on the sensitive dataset using DP‑SGD with LoRA, 
    while all remaining parameter matrices are kept frozen.%
    }
    \label{fig:sapf}
\end{figure} 

\vspace{1mm}
\noindent \textbf{DP Image Synthesis.}
By creating synthetic datasets that mimic real-world distributions while satisfying DP, DP image synthesis enables organizations to share high-quality synthetic data without compromising individual privacy.

This work focuses on DP image synthesis methods using public models. \Cref{tab:SumOfMethods} presents current \textit{state-of-the-art} popular methods. Previous 
work neither finetunes all parameters of public models~\cite{dp-diffusion} nor follows the fine‑tuning practices commonly used in non‑DP image synthesis~\cite{peft,nondpfinetuning}; instead, it fine‑tunes only the attention layers~\cite{dpldm,dplora}. Unlike these fixed-layer strategies, our \toolname{} introduces a saliency-aware selection mechanism that dynamically identifies the critical
parameters based on sensitive images under DP.

\begin{table}[!t]
\renewcommand{\arraystretch}{1.1}
\setlength{\tabcolsep}{10pt}
    \centering
    \caption{Comparison of existing DP image synthesis methods leveraging public models. The first two methods listed above the horizontal line are fine-tuning-free, whereas the remaining methods use fine-tuning.}
    \label{tab:SumOfMethods}
    \resizebox{0.48\textwidth}{!}{
    \begin{tabular}{l|ccc}
    \toprule
    Method & Training Methods & Parameter Selection & Year\\
    \hline
    PE~\cite{dpsda} & Fine-Tuning Free & - & 2024\\
    Aug-PE~\cite{xiedifferentially} & Fine-Tuning Free & - & 2024 \\
    \hline
    DP-Finetune~\cite{dp-diffusion} & DP-SGD & Full Parameter & 2023 \\
    DP-LDM~\cite{dpldm} & DP-SGD & Attention & 2024 \\
    DP-LoRA~\cite{dplora} & DP-SGD + LoRA & Attention & 2025 \\
    \cellcolor{gray!20}{\toolname{} (Ours)} & \cellcolor{gray!20}{DP-SGD + LoRA} & \cellcolor{gray!20}{Saliency-Aware Selection} & \cellcolor{gray!20}{2026} \\
    \bottomrule
\end{tabular}}
\end{table}

\subsection{Diffusion Model}
\label{subsec:dm}

Recent advancements in diffusion models have significantly enhanced image generation capabilities, leading current DP image synthesis methods to primarily adopt diffusion models as the synthesizers~\cite{dp-diffusion,dp-feta,li2023privimage,dpsda,pe3}. 

\vspace{1mm}
\noindent \textbf{Traditional Diffusion Models.} Diffusion models~\cite{ddpm} consist of two complementary stages:
\begin{itemize}[leftmargin=*]
    \item \emph{Forward diffusion}, which incrementally perturbs a clean image $x_0$ by adding Gaussian noise over $T$ steps, producing a sequence $\{x_i\}_{i=1}^T$ that transfers to pure noise.
    \item \emph{Reverse diffusion}, which iteratively removes noise to recover a clean image from random initialization.
\end{itemize}
In the forward processes, the transition between consecutive noisy images, denoted as $p(x_t | x_{t-1})$, follows a multi-dimensional Gaussian,
\begin{equation}
    p(x_t | x_{t-1}) = \mathcal{N}\left( x_t; \sqrt{1 - \beta_t} x_{t-1}, \beta_t \mathbb{I} \right),
    \label{eq:variant}
\end{equation}
where $\beta_t$ is a hyperparameter that controls the noise variance at step $t$. Defining $\bar{\alpha}_t = \prod_{s=1}^t (1 - \beta_s)$, the marginal distribution of $x_t$ given $x_0$ is, $p(x_t | x_0) = \mathcal{N}\left( x_t; \sqrt{\bar{\alpha}_t} x_0, (1 - \bar{\alpha}_t) \mathbb{I} \right)$. This allows direct sampling of $x_t$ from $x_0$ via, ${x_t} = \sqrt {{{\bar \alpha }_t}} {x_0} + e_t\sqrt {1 - {{\bar \alpha }_t}} ,e_t\sim \mathcal{N}\left( {0,\mathbb{I}} \right).$
The objective of diffusion models learn a denoising network $e_\theta(x_t, t)$ that predicts the noise $e_t$ added at each step~\cite{ddpm},
\begin{equation} \label{eq:L_DM} \mathcal{L} = \mathbb{E}_{x_0 \sim D,\, t \sim U\{1, T\},\, e_t \sim \mathcal{N}(0, \mathbb{I})} \left\| e_t - e_\theta(x_t, t) \right\|_2^2.
\end{equation}
The $D$ is the dataset of images, and $U\{1,T\}$ means the uniform distribution over time steps. Once trained, $e_\theta$ enables image synthesis by progressively denoising Gaussian noise.

\begin{figure}[!t]
    \centering
    \setlength{\abovecaptionskip}{0pt}
    \includegraphics[width=1.0\linewidth]{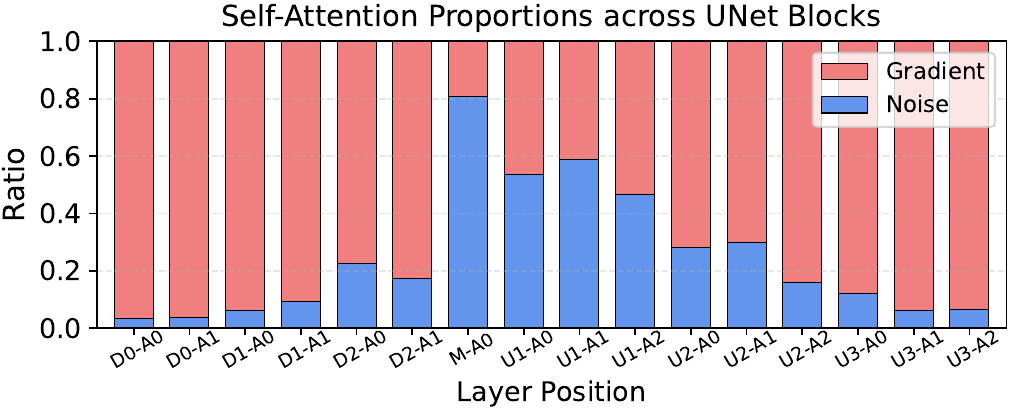}
    \caption{\gc{Each cell shows the ratio of noise scale and gradient norm for self-attention on {\tt CIFAR-10}. The public model is `Stable-Diffusion-v1-5'. `D$i$-A$j$' = Down block $i$, attention layer $j$;  `U$i$-A$j$' = Up block $i$, attention layer $j$.  `M-A0' = Middle block, attention layer 0.}%
    }
    \label{fig:sapf_ratio}
\end{figure} 

\noindent \textbf{Latent Diffusion Models.} Latent diffusion models~\cite{labelembedding} are a variant of diffusion models that operate in a latent space instead of in pixel space. This approach uses a pretrained autoencoder consisting of an encoder $\mathcal{E}$ and a decoder $\mathcal{R}$. Specifically, the encoder $\mathcal{E}$ maps a clean image $x_0$ into a lower-dimensional latent representation $z_0 = \mathcal{E}(x_0)$. The forward diffusion process, consistent with the logic in~\Cref{eq:variant}, is then applied to $z_0$ to produce a sequence of latent variables $\{z_t\}_{i=1}^T$. The latent variable at step $t$ can be expressed as, $z_t = \sqrt{\bar{\alpha}_t} z_0 + e_t \sqrt{1 - \bar{\alpha}_t}, \; e_t \sim \mathcal{N}(0, \mathbb{I}).$
The denoising network $e_\theta(z_t, t, c)$ is then trained to predict the noise $e_t$ added to the latent vector, often incorporating a conditioning vector $c$ (e.g., text or class labels) via an attention mechanism~\cite{vaswani2017attention}. The training objective is,
\begin{equation} 
\label{eq:L_LDM}\mathcal{L}_{\text{LDM}} = \mathbb{E}_{{(x_0,c)\sim D, t \sim U\{1, T\}, e_t \sim \mathcal{N}(0, \mathbb{I})}} \left\| e_t - e_\theta(z_t, t, c) \right\|_2^2.
\end{equation}
After training, a latent sample $z_T$ is drawn from a Gaussian distribution and iteratively denoised to recover $z_0$. Finally, the high-resolution image is reconstructed by passing the denoised latent through the decoder, $\tilde{x} = \mathcal{R}(z_0)$.

By operating in this compressed space, latent diffusion models significantly reduce the training and sampling costs while preserving high perceptual quality. This efficiency makes them a preferred backbone for recent DP synthesis methods, such as DP-LDM~\cite{dpldm} and DP-LoRA~\cite{dplora}.

\subsection{Threat Model}
\label{subsec:threat_model}

We assume data providers hold sensitive images, e.g., medical images, and sharing them directly poses privacy risks. Several approaches propose generating datasets to replace real ones~\cite{realfake}. However, this method does not fully mitigate privacy concerns, as adversaries still infer sensitive information using synthetic images. \toolname{} enforces image-level DP, a \textit{general protection} against various inference attacks targeting individual images, such as membership inference~\cite{2023whithmiadiffusion}. This paradigm has been widely explored in prior works~\cite{pe3,xiedifferentially,gong2025dpimagebench,dp-kernel,lin2020using,yin2022practical,wang2025synthesize,wang2025struct}, and we follow the same approach.

In principle, one solution for DP image synthesis is to directly pre‑train a model entirely on public datasets before applying DP fine‑tuning. This has the advantage of eliminating any pretraining–finetuning overlap concerns (as discussed in~\cite{position}), but it is computationally expensive and slow~\cite{gong2025dpimagebench}. On the other hand, there already exist many high‑quality public models released on open platforms. Therefore, \toolname{} follows the more practical path of directly leveraging these public models. Although pretraining data might partially overlap with sensitive data, we mitigate this concern by evaluating across multiple public models and conducting all comparisons within each individual model. This ensures a fair evaluation.

\section{Methodology}
\label{sec:ass}

This section elaborates on our methods, \toolname, including the motivation, technical details, and privacy analysis.

\subsection{Motivation}
\label{subsec:mot}

Existing works~\cite{dplora,dpldm} restrict fine-tuning to attention layers, a heuristic inherited from non-DP image synthesis methods~\cite{peft,nondpfinetuning}. \gc{In contrast, \toolname{} posits that exhaustive fine-tuning of all attention layers is suboptimal under DP, as such extensive parameter coverage exacerbates noise accumulation, making training collapse. To the best of our knowledge, this issue has not been explicitly addressed in prior DP image synthesis methods~\cite{gong2025dpimagebench,dplora,dpldm,pe3,li2023privimage}. Specifically, as shown in~\Cref{fig:sapf_ratio}, the attention layers of the middle block~\cite{labelembedding} exhibit a high ratio of noise scale and gradient
norm.} These layers are essential in processing feature refinement and text-to-image alignment. Under DP, minor parameter updates can steer generative performance. As shown in~\Cref{tab:moti}, removing fine-tuning of the mid‑block attention layer leads to a substantial improvement in synthetic performance. However, identifying sensitive layers is time‑consuming, as public models typically contain a large number of parameters and complex architectures. \gc{Even worse, this strategy is not DP‑compliant, as manual parameter selection relies on repeated inspection of sensitive data~\cite{koskela2023practical}. This motivates us to use gradient-based metrics to analytically select parameters that improve DP fine-tuning.}

\textit{The first question is how to select the parameters.} The motivation for using gradient as a selection saliency metric is primarily rooted in the Signal-to-Noise Ratio (SNR) during DP training~\cite{Li2022largeDP}. In DP-SGD, a fixed noise is injected into all trainable parameters to satisfy DP. Parameters with small gradients possess a low SNR, meaning their learning signal is easily drowned out by the noise. In contrast, parameters with large gradient magnitudes exhibit higher noise resilience, as their signal is strong enough to remain impactful even after perturbation. By selecting these high-gradient parameters, \toolname{} ensures that the optimization is concentrated on a high-SNR subspace, where the learning gradient dominates the noise.
\gc{\Cref{sec:rq2} shows that \toolname{} can adaptively avoid fine-tuning low SNR parameters in the public model under DP, preventing the model from being derailed by the accumulation of noise in low-signal components.}

\textit{The second question concerns the structural granularity of the selection.} 
Whereas conventional approaches typically finetune entire layers~\cite{dpldm,labelembedding}, 
\toolname{} instead operates at a matrix-wise granularity.  This design choice is grounded in the functional heterogeneity within attention modules. 
Recent work~\cite{attentionlayerimproces} shows that the query, key, and value matrices contribute unequally to fine-tuning performance, with certain matrices (e.g., value metrics in attention~\cite{labelembedding}) exerting substantially greater influence than others. Thus, matrix-wise selection enables a more favorable balance between model expressivity and the distortions introduced by DP.

Experiments conducted in~\Cref{sec:rq2} shows that matrix-wise selection achieves better synthetic performance than layer-wise selection, on studied sensitive image datasets.

\begin{table}[!t]
\small
    \centering
    \caption{ Acc (\%) and FID of synthetic images for {\tt CIFAR-10}, using `Stable-Diffusion-v1-5' as the public model, under $\epsilon=10$. `DP-LoRA w/o Middle' means using DP-LoRA to finetune the public model but excluding the middle block. `LoRA w/o DP' means a non-DP LoRA fine-tuning. }
    \label{tab:moti}
    \setlength{\tabcolsep}{1.2mm}{
    \resizebox{0.48\textwidth}{!}{
    \renewcommand{\arraystretch}{1.2}
    \begin{tabular}{l|cccc}
    \toprule
    \textbf{Metrics}   & DP-LoRA & DP-LoRA w/o Middle  & LoRA w/o DP   & DP-SAPF \\
    \hline
    Acc & 13.3 & 69.7 &  82.4  & 74.6 \\
    FID & 384.3 & 31.3 &  12.9  & 24.6 \\
    \bottomrule
\end{tabular}
}}
\end{table}

\subsection{Saliency-Aware Parameter Selection}
\label{subsec:sapf}

As introduced in~\Cref{subsec:mot}, \toolname{} selects a subset of salient weight matrices from the public model by analyzing the gradients induced by sensitive images. Following prior works~\cite{peft,dpldm,dplora}, \toolname{} further restricts matrix-wise selection to the attention layers. Their parameters exhibit strong responsiveness to task‑relevant information while maintaining reliable training behavior, making them an effective and robust subset for adaptation~\cite{li2025lorada}. Fine-tuning feed‑forward or normalization layers tends to either achieve marginal improvements or introduce training instability~\cite{peft}. 

Each attention layer comprises a set of projection matrices that parameterize the query, key, and value transformations, denoted as $\mathbf{W}^q,\ \mathbf{W}^k,\ \mathbf{W}^v$. These head-specific projections are usually followed by an output projection that linearly combines the concatenated head outputs~\cite{labelembedding}. Our selection targets these three weight matrices in attention layers. 

\begin{algorithm}[!t]
       \caption{Saliency-Aware Parameter Selection.}
      \label{alg:selection}
       \SetKwInOut{Input}{Input}
      \SetKwInOut{Output}{Output}
      \SetKwFunction{SpaticalFeatureQuery}{SpaticalFeatureQuery}
      \SetKwProg{Fn}{Function}{:}{}
    \Input{Sensitive dataset $D_s$ with size $N$ and estimated size $N^\ast$; the public model $M_{\text{pub}}$; the noise scale $\sigma_s$; the weight matrices candidate $\Theta_{\text{full}} = \{\mathbf{W}_i\}_{i=1}^K$.}
    \tcp{{\color{gray} Gradient Calculation}}
    Init clip gradient set $G = \varnothing$; \\
    \While{$x_i \in D_s$}{
        Calculate gradient for $x_i$ on $M$ using~\Cref{eq:L_LDM} and obtain $g_i = [\text{vec}(\nabla_{\mathbf{W}_1} \mathcal{L}_i), \dots, \text{vec}(\nabla_{\mathbf{W}_K} \mathcal{L}_i)]$;\\
        Using~\Cref{eq:joint_clip} to clip the $g_i$ and obtain $\bar{g}_i = \left[\Gamma_{C_s}(g_i)\big|_{\mathbf{W}_1}, \cdots, \Gamma_{C_s}(g_i)\big|_{\mathbf{W}_K}\right]$;\\
        $G = G \cup \bar{g}_i$; \\
    }
    Init average clipped gradient set $S = \varnothing$; \\
    \While{$\textbf{W}_i \in \Theta$}{
        Calculate average clipped gradient $S_k(D_s)=\frac{1}{N^\ast}\sum_{i=1}^{N}\bar g_{i,k},\; \text{where,} \; \bar g_{i,k} = \Gamma_{C_s}(g_i)\big|_{\mathbf{W}_k}$; \\
        Obtain $\tilde{G}_{s,k}$ by using~\Cref{eq:noise_inject}; \\
        $S = S \cup \tilde{G}_{s,k}$; \\
    }
    \tcp{{\color{gray} Parameter Selection}}
     Based on the norms of the noisy gradient $S =[\tilde{G}_{s,1},\cdots,\tilde{G}_{s,K}]$, we select the top‑$c$ matrices:
     $\Theta' = \operatorname{arg\,top\mbox{-} }c \left( \left\{ \left\| \tilde{G}_{s,k} \right\| \right\}_{k=1}^{K} \right);$ \\
    \Output{The saliency-aware parameter matrix set $\Theta'$.}
\end{algorithm}

Let the public model consist of $K$ weight matrices $\Theta_{\text{full}} = \{\mathbf{W}_1, \mathbf{W}_2, \dots, \mathbf{W}_K\}$ that we aim to select. For each sample in the sensitive image dataset $x_i \in D_s, D_s = \{x_i\}_{i=1}^N$, we perform a forward pass to compute the task loss $\mathcal{L}_{\text{LDM}}$, as defined in~\Cref{eq:L_LDM}, followed by backpropagation to obtain the gradients with respect to every candidate matrix, $g_{i} = \left.\left\{\nabla_{\mathbf{W}_j} \mathcal{L}_i \right| j=1,\cdots, K\right\}$. To bound sensitivity, we apply a joint clipping operator $\Gamma_{C_s}$ to the concatenated gradient vector $g_i = [\text{vec}(\nabla_{\mathbf{W}_1} \mathcal{L}_i), \dots, \text{vec}(\nabla_{\mathbf{W}_K} \mathcal{L}_i)]$, where $\text{vec}(\cdot)$ denotes the flattening operator. The joint clipping operator $\Gamma_{C_s} (g_i)$  calculates the average gradient $\bar{g}_i = \Gamma_{C_s} (g_i)$ as follows,
\begin{equation}
\label{eq:joint_clip}
    \bar{g}_i =  g_i \left/ \max\left(1, \frac{\|g_i\|_2}{C_s}\right)\right., \|g_i\|_2 = \sqrt{\sum_{k=1}^K \|\text{vec}(\nabla_{\mathbf{W}_k} \mathcal{L}_i)\|^2}.
\end{equation}
The clipped gradient $\bar g_i$ is then partitioned back into its per‑matrix components
$\bar g_{i,k} = \Gamma_{C_s}(g_i)\big|_{\mathbf{W}_k}$. For each matrix $\mathbf{W}_k \in \Theta_{\text{full}}$, we compute the average clipped gradient $S_k(D_s)$ over the sensitive dataset $D$ of size $N$,
\begin{equation}
\label{eq:noise_inject}
    S_k(D_s)=\frac{1}{N^\ast}\sum_{i=1}^{N}\bar g_{i,k},%
\end{equation}
where $N^*$ approximates sensitive image dataset size $N$ as introduced in Appendix~\ref{apsec:proof}. The private selection criterion is based on the noisy average gradient,
\begin{equation}
    \tilde{G}_{s,k} = S_k(D_s) + \mathcal{N}\left(0, \sigma_s^2 \Delta_s^2 \mathbb{I}\right),
\end{equation}
where $\sigma_s^2$ is a hyperparameter controlling the scale of the injected Gaussian noise, and the global sensitivity is $\Delta_s = C_s / N^\ast$. In \Cref{the:select}, we prove this procedure satisfies DP.

\begin{theorem}
\label{the:select}
    The averaged clipped gradient query $[S_1(D_s),...,S_K(D_s)]$ has global $\ell_2$ sensitivity $\Delta_s = C_s / N^\ast.$
    For any Rényi order $\alpha > 1$, adding Gaussian noise $\mathcal{N}\!\left(0, \sigma_s^2 \Delta_s^2 \mathbb{I}\right)$ to $S_k(D_s)$ for the weight matrices $\{\mathbf{W}_1, \dots, \mathbf{W}_K\}$ ensures that the resulting mechanism satisfies $(\alpha, \gamma_s)$-RDP for $\gamma_s = \frac{\alpha}{2\sigma_s^2}$.
\end{theorem}

\noindent We provide the proof of \Cref{the:select} in~\Cref{apsec:proof}. Based on the norms of the noisy gradient $[\tilde{G}_{s,1},\cdots,\tilde{G}_{s,K}]$, we select the top‑$c$ matrices for fine‑tuning. Formally, the selected parameter subset is defined as,
$$\Theta' = \operatorname{arg\,top\mbox{-} }c \left( \left\{ \left\| \tilde{G}_{s,k} \right\| \right\}_{k=1}^{K} \right),$$
where $c$ denotes the ratio specifying how many matrices with the largest noisy gradient norms are selected for fine‑tuning.

\Cref{alg:selection} describes the workflow of our saliency‑aware parameter selection mechanism, which determines the parameter matrices to be fine‑tuned on sensitive data.

\subsection{DP Fine-Tuning}

Building on prior work on public model fine-tuning (both DP and non-DP image synthesis)~\cite{peft,dplora,labelembedding,nondpfinetuning}, this work adopts LoRA as an efficient adaptation mechanism under DP, given the large parameter counts of public models. Fine-tuning all parameters is time-consuming. LoRA~\cite{lora} trains a public model by injecting low-rank trainable updates into its weight matrices, instead of optimizing all model parameters.
By limiting the number of trainable parameters, LoRA improves computational efficiency and facilitates privacy-preserving optimization, while maintaining the expressive power of the public model. 
The following theorem shows that given a fixed clipping threshold $C$ and noise multiplier $\sigma_d$, DP-SGD applied to full-parameter fine-tuning and LoRA fine-tuning satisfy the same level of $(\epsilon, \delta)$-DP.
\begin{theorem}
\label{the:dplora}
    Given a fixed clipping threshold $C$, noise multiplier $\sigma_d$, sampling rate $q$, and number of iterations $t_d$, the $(\epsilon, \delta)$-DP guarantee of DP-SGD remains identical regardless of whether the training is performed on the full parameter set $\Theta$ or a low-rank subspace $\Theta_{\text{LoRA}}$.
\end{theorem}

\noindent We provide the proof of~\Cref{the:dplora} in Appendix~\Cref{apsec:proof}. Based on the saliency-aware selection procedure introduced in~\Cref{subsec:sapf}, we obtain a subset
$\Theta' = \{\mathbf{W}_{1}, \dots, \mathbf{W}_{H}\} \subseteq \Theta_{\text{full}},$ which contains the $H = c\times K$ matrices. LoRA updates are then applied \emph{only} to matrices in $\Theta'$. Formally, for each selected matrix $\mathbf{W}_k \in \Theta'$ with shape $m_k \times h_k$, where $m_k$ and $h_k$ are the dimensions of the weight matrix $\mathbf{W}_k$, LoRA represents the adapted weight as,
\begin{equation}
\label{eq:lora}
\mathbf{W}'_k = \mathbf{W}_k + \Delta \mathbf{W}_k, \quad \Delta \mathbf{W}_k = \mathbf{A}_k\mathbf{B}_k,
\end{equation}
where $\mathbf{A}_k \in \mathbb{R}^{m_k \times r}$ and $\mathbf{B}_k \in \mathbb{R}^{r \times h_k}$ are low-rank matrices, and $r \ll \min(m_k, h_k)$ denotes the rank of the decomposition. $\mathbf{A}_k$ and $\mathbf{B}_k$ are initialized asymmetrically: $\mathbf{A}_k$ is typically sampled from a Gaussian distribution, while $\mathbf{B}_k$ is initialized to zero~\cite{lora}.
This formulation reduces the number of trainable parameters from $m_k \times h_k$ to $r \times (m_k + h_k)$. All original weights in $\Theta_{\text{full}}$ remain frozen. The set of trainable parameters is restricted to the adapters associated with the salient subset, $\Theta_{\text{train}} = \{(\mathbf{A}_k, \mathbf{B}_k) \mid \mathbf{W}_k \in \Theta'\}.$

In \toolname, during each training iteration, for every sample $x_i$ in a sensitive sub-batch 
$D_s^{\text{sub}} = \{x_i\}_{i=1}^{B}$, we compute the gradients with respect to all trainable LoRA
parameters associated with the selected matrices $\Theta'$. 
For each selected weight matrix $\mathbf{W}_k \in \Theta'$, LoRA introduces a pair of low‑rank parameters 
$(\mathbf{A}_k,\mathbf{B}_k)$
Given the latent diffusion objective in \Cref{eq:L_LDM}, we compute the per‑sample gradients, $
(\nabla_{\mathbf{A}_k}\mathcal{L}_i, \nabla_{\mathbf{B}_k}\mathcal{L}_i), \forall\, \mathbf{W}_k \in \Theta'.$

\begin{algorithm}[!t]
       \caption{The workflow of \toolname.}
      \label{alg:dpsapf}
       \SetKwInOut{Input}{Input}
      \SetKwInOut{Output}{Output}
      \SetKwProg{Fn}{Function}{:}{}
    \Input{Public model $M$ with parameter matrices $\Theta'$; the number of matrix in $\Theta'$, $H$; sensitive dataset $D_s$; learning rate $\lambda$; finetuning batch size $B$ and estimated size $B^\ast$; the noise scale $\sigma_d$; the clip bound $C$; finetuning iteration $t_d$.}
    \tcp{{\color{gray} Private Fine-tuning}}
    Init $h=0$; \\
    Init $\Theta_{\text{train}} = \{(\mathbf{A}_k, \mathbf{B}_k) \mid \mathbf{W}_k \in \Theta', k=[1,H]\}$, where $A_k \sim \mathcal{N}(0, \mathbb{I})$ and $\mathbf{B}_k$ is initialized to zero; \\
    \While{$h < t_d$}{
        Sample subset $D_{s}^{\text{sub}}=\{x_i\}_{i=1}^{B}$ from $D_s$; \\
        Init gradient set $G = \varnothing$;\\
        \For{$ x_i \in D_{s}^{\text{sub}}$}{
            Calculate the gradient for $x_i$, using \Cref{eq:L_LDM}, $(\nabla_{\mathbf{A}_k}\mathcal{L}_i, \nabla_{\mathbf{B}_k}\mathcal{L}_i), \forall\, \mathbf{W}_k \in \Theta';$ \\
            Calculate the aggregate using $g_i = \bigoplus_{\mathbf{W}_k \in \Theta'} \left[\mathrm{vec}\!\left(\nabla_{\mathbf{A}_k}\mathcal{L}_i\right)\oplus\mathrm{vec}\!\left(\nabla_{\mathbf{B}_k}\mathcal{L}_i\right)\right];$ \\
            Using~\Cref{eq:dplora_clip} to clip $g_i$ and obtain $\bar{g}_i$; \\
            $G = G \cup \bar{g}_i$; \\
        }
        Calculate the privatized aggregated gradient, $g_p = \lambda \left( \frac{1}{{B^*}} \sum_{\bar{g}_i \in G} \bar{g}_i + \frac{C}{{B^*}} \mathcal{N}(0, \sigma_d^2 \mathbb{I}) \right)$; \\
        The $g_p$ is  used to update $(\mathbf{A}_k,\mathbf{B}_k)$ for all $\mathbf{W}_k \in \Theta'$; \\
        $h = h + 1$; \\
    }
    Obtain $\mathbf{W}_k'$ for all $ \mathbf{W}_k \in \Theta'$ using~\Cref{eq:lora}; \\
    \Output{The well-finetuned public model $M'$.}
\end{algorithm}

To satisfy DP, we construct an aggregate gradient vector by vectorizing and concatenating the gradients of all selected matrix adapters, which is defined as follows,
\begin{equation} \label{eq:aggregate_grad} 
    g_i = \bigoplus_{\mathbf{W}_k \in \Theta'} 
    \left[
        \mathrm{vec}\!\left(\nabla_{\mathbf{A}_k}\mathcal{L}_i\right)
        \oplus
        \mathrm{vec}\!\left(\nabla_{\mathbf{B}_k}\mathcal{L}_i\right)
    \right],
\end{equation}
where $\mathrm{vec}(\cdot)$ denotes the vectorization operator and $\oplus$ signifies concatenation. To bound the global $\ell_2$-sensitivity of the model update, the combined gradient vector $g_i$ is constrained by a clipping threshold $C$,
\begin{equation}
\label{eq:dplora_clip}
    \bar{g}_i = g_i \left/ \max\left(1, \frac{\|g_i\|_2}{C}\right)\right.,
\end{equation}
where $\|g_i\|_2^2 
= \sum_{\mathbf{W}_k \in \Theta'}
\left(
    \|\mathrm{vec}(\nabla_{\mathbf{A}_k}\mathcal{L}_i)\|_2^2
    +
    \|\mathrm{vec}(\nabla_{\mathbf{B}_k}\mathcal{L}_i)\|_2^2
\right)$. Following the clipping step, the gradients are averaged across the sub-batch and perturbed with Gaussian noise as defined in~\Cref{eq:dpsgd} to satisfy the $(\epsilon, \delta)$-DP requirement. These privatized gradients are then used to update only the
LoRA parameters $(\mathbf{A}_k,\mathbf{B}_k)$ for all $\mathbf{W}_k \in \Theta'$. After private training, $(\mathbf{A}_k,\mathbf{B}_k)$ are incorporated into the effective weights $\textbf{W}_k'$ in~\Cref{eq:lora} for all $\mathbf{W}_k \in \Theta'$.

Detailed text-processing is provided in \Cref{apsubsec:text}, and the workflow of \toolname{} is summarized in Algorithm~\ref{alg:dpsapf}.

\begin{table}[!t]
\footnotesize
    \centering
    \caption{Summary of sensitive image datasets, including resolution, data split statistics, and the number of categories.}
    \label{tab:datainfo}
    \resizebox{0.48\textwidth}{!}{
    \begin{tabular}{l|ccccc}
    \toprule
    \textbf{Datasets} & \textbf{Training} & \textbf{Validation} & \textbf{Test} & \textbf{Resolution} & \textbf{Category} \\
    \hline
    {\tt CIFAR-10} & 45,000& 5,000 & 10,000 & 32$\times$32 & 10\\
    {\tt OCTMNIST} & 97,477 & 10,832 & 1,000 & 128$\times$128 & 4 \\
    {\tt CelebA} & 162,770 & 19,867 & 19,962 & 256$\times$256 & 2\\
    {\tt Camelyon} & 302,436 & 34,904  & 85,054 & 96$\times$96 & 2\\
    \bottomrule
\end{tabular}
}
\end{table}

\subsection{Privacy Analysis}
\label{subsec:privacy_analysis}

In \toolname{}, two components consume the privacy budget: (1) the parameter-selection stage that computes per-sample gradients, clips them to the $\ell_2$ bound $C_s$, and adds Gaussian noise; and (2) the fine-tuning stage that trains the selected parameters using DP-SGD. The parameter-selection stage is a single Gaussian mechanism that satisfies $(\alpha, \gamma_s)$-RDP. During DP-SGD, each iteration samples a sub-batch with ratio $q$ and injects Gaussian noise with scale $\sigma_d$, which corresponds to an SGM~\cite{sgm}, as introduced in~\Cref{app:supp_dp}. According to the RDP analysis of SGM~\cite{sgm,dpsgd}, the $t_d$ fine-tuning iterations incur an RDP cost $(\alpha,\gamma_d(\sigma_d))$, where $\gamma_d$ depends on $q$, $t_d$, and $\sigma_d$. Because RDP composes linearly~\cite{rdp}, the overall privacy guarantee of \toolname{} is $(\alpha, \gamma_s + \gamma_d)$. 

To better understand how the privacy budget is distributed across the two stages, we define the privacy budget ratios as, $r_s = \gamma_s/(\gamma_s + \gamma_d), \;
r_d = \gamma_d/(\gamma_s + \gamma_d), $
which measures the fraction of the total RDP cost consumed by the parameter-selection stage and the DP-SGD fine-tuning stage, respectively. These ratios quantify the extent to which each stage contributes to overall privacy protection.

To ensure that \toolname{} satisfies a target $(\varepsilon,\delta)$-DP, we choose the privacy parameters in three steps. (1) We first set the noise scale $\sigma_s$ to determine the RDP cost $(\alpha,\gamma_s)$ for the feature-query stage. (2) We fix the number of fine-tuning iterations $t_d$ and the sampling ratio $q$, under which the RDP cost of DP-SGD becomes a function of the noise scale $\sigma_d$, denoted by $(\alpha,\gamma_d(\sigma_d))$. (3) We search over different values of $\sigma_d$, convert the composed RDP cost $(\alpha, \gamma_s+\gamma_d(\sigma_d))$ to its corresponding $(\varepsilon,\delta)$-DP guarantee~\cite{rdp}, and select the smallest $\sigma_d$ that satisfies the given privacy budget. The definitions and more details of RDP in \toolname{} are provided in~\Cref{app:supp_dp}.

\begin{table}[!t]
\footnotesize
    \centering
    \caption{Summary of public models used in our experiments.}
    \label{tab:publicdatainfo}
    \resizebox{0.48\textwidth}{!}{
    \begin{tabular}{l|ccccc}
    \toprule
    \textbf{Public Model} & \textbf{Source} & \textbf{Resolution} & \textbf{Size} & \textbf{Year} \\
    \hline
     Stable-Diffusion-v1-5\textsuperscript{\ref{sdv15}} & Stability AI & 512$\times$512 & 1B & 2022 \\
     Stable-Diffusion-2-1-base\textsuperscript{\ref{sd21base}} &  Stability AI&  512$\times$512 & 1B & 2022 \\
     Realistic-v6\textsuperscript{\ref{realistic}} &  Hugging Face& 896$\times$896 & 1B & 2024 \\
     Prompt2med\textsuperscript{\ref{promptmed}} & Hugging Face & 512$\times$512 & 1B & 2024 \\
    \bottomrule
\end{tabular}
}
\end{table}

\Cref{sec:rq3} shows that under our default configuration ($\sigma_s=5, \epsilon=10$), the privacy overhead of the parameter-selection stage accounts for only 0.17\% of the total privacy budget. This indicates that the parameter selection is highly budget-efficient, preserving the vast majority of the privacy allowance for the subsequent fine-tuning phase.
\section{Experimental Setup}
\label{sec:setup}

\begin{table*}[!t]
\renewcommand{\arraystretch}{1.1}
\setlength{\tabcolsep}{8.5pt}
\small
\centering
\caption{FID and Acc (\%) of synthetic images generated by different public models with \toolname{} and baselines under $\epsilon=\{1,10\}$, on {\tt CIFAR-10}, {\tt OCTMNIST}, {\tt CelebA}, and {\tt Camelyon}. The best values are highlighted in bold in each column. }
\label{tab:rq1}
\resizebox{1.0\textwidth}{!}{
\begin{tabular}{l|>{\centering\arraybackslash}l|>{\centering\arraybackslash}c>{\centering\arraybackslash}c|>{\centering\arraybackslash}c>{\centering\arraybackslash}c|>{\centering\arraybackslash}c>{\centering\arraybackslash}c|>{\centering\arraybackslash}c>{\centering\arraybackslash}c|>{\centering\arraybackslash}c>{\centering\arraybackslash}c|>{\centering\arraybackslash}c>{\centering\arraybackslash}c|>{\centering\arraybackslash}c>{\centering\arraybackslash}c|>{\centering\arraybackslash}c>{\centering\arraybackslash}c}
\toprule
\multirow{3}{*}{\textbf{Public Model}} & \multirow{3}{*}{\textbf{Method}}  & \multicolumn{4}{c|}{{\tt CIFAR-10}} & \multicolumn{4}{c|}{{\tt OCTMNIST}} & \multicolumn{4}{c|}{{\tt CelebA}} & \multicolumn{4}{c}{{\tt Camelyon}} \\
\cline{3-18}
& & \multicolumn{2}{c|}{\textbf{$\epsilon=1$}} & \multicolumn{2}{c|}{\textbf{$\epsilon=10$}} & \multicolumn{2}{c|}{\textbf{$\epsilon=1$}} & \multicolumn{2}{c|}{\textbf{$\epsilon=10$}}  & \multicolumn{2}{c|}{\textbf{$\epsilon=1$}} & \multicolumn{2}{c|}{\textbf{$\epsilon=10$}} & \multicolumn{2}{c|}{\textbf{$\epsilon=1$}} & \multicolumn{2}{c}{\textbf{$\epsilon=10$}} \\
\cline{3-18}
  & & \textbf{FID} & \textbf{Acc} & \textbf{FID} & \textbf{Acc} & \textbf{FID} & \textbf{Acc} & \textbf{FID} & \textbf{Acc} & \textbf{FID} & \textbf{Acc} & \textbf{FID} & \textbf{Acc} & \textbf{FID} & \textbf{Acc} & \textbf{FID} & \textbf{Acc} \\
\midrule
\multirow{6}{*}{SD-v1-5} & PE & \textbf{16.4} & 67.5 & \textbf{10.9} & 66.9 &  132.5 & 27.1 & 104.2 & 26.1 & 44.6 & 68.0  & 40.4 & 64.1 & \textbf{62.3} & 69.7 & 55.9 & 50.4 \\
& Aug-PE & 24.0 & 44.8 & 12.9 & 49.8 & 143.2 & 28.4 & 92.1 & 30.2 & 41.1 &  70.2 & 42.5 & 65.6 & 63.1  & 70.1 & \textbf{51.9} & 64.9 \\
& DP-LDM & 240.4 & 13.0 & 183.8 & 14.6 & 398.2 & 25.0 & 347.8 & 25.0 & 312.8  & 57.3 & 281.6  & 61.0& 405.8 & 60.4  & 248.8 & 69.8 \\
& DP-LoRA & 366.0 &  12.8 & 384.3 & 13.3 & 320.3 &  25.0 & 347.0 &  25.0 & 244.7 &  61.4 & 447.7 &  66.3 & 513.7 &  56.3 & 437.1 &  62.1 \\
& DP-Finetune & 191.9 & 15.0 & 145.2 & 14.2 & 397.8 & 31.9 & 345.2 & 30.9 & 277.9 & 60.8 & 274.6 & 61.7 & 242.8  & 57.6 & 187.6 & 65.3 \\
\cline{2-18}
& \cellcolor{gray!20}{\toolname} & \cellcolor{gray!20}{32.1} & \cellcolor{gray!20}{\textbf{73.1}} & \cellcolor{gray!20}{26.6} & \cellcolor{gray!20}{\textbf{74.6}} & \cellcolor{gray!20}{\textbf{85.4}} & \cellcolor{gray!20}{\textbf{44.6}} & \cellcolor{gray!20}{\textbf{77.9}} & \cellcolor{gray!20}{\textbf{46.2}} & \cellcolor{gray!20}{\textbf{25.0}} & \cellcolor{gray!20}{\textbf{84.2}} & \cellcolor{gray!20}{\textbf{23.6}} & \cellcolor{gray!20}{\textbf{90.2}} & \cellcolor{gray!20}{139.0} & \cellcolor{gray!20}{\textbf{72.5}} & \cellcolor{gray!20}{145.0} & \cellcolor{gray!20}{\textbf{79.9}} \\
\hline
\multirow{6}{*}{SD-2-1-base} & PE & \textbf{17.4} & 66.3 & 12.1 & 68.5 & 135.6 & 25.3 & 112.1 & 27.4 & 42.8 &  70.2 & 36.8 & 73.0 &  \textbf{68.4} & 66.6 & 60.4 & 60.9 \\
& Aug-PE & 25.6 & 61.3 & \textbf{11.1} & 66.9 & 132.4 & 26.0 & 110.3 & 28.1 & 40.2 & 73.2  & 38.3 & 74.7 &  69.4 & 65.2 & 63.4 & 62.3 \\
& DP-LDM & 119.4 & 15.3 & 142.1 & 16.1 & 304.5 & 25.0 & 317.9 & 26.3 & 299.3 & 60.1  & 220.5 & 63.1 &  388.8 & 52.3 & 300.0 & 62.7 \\
& DP-LoRA & 301.1 & 11.7 & 289.4 & 14.9 & 377.3 & 25.0 & 333.5 & 25.0 & 231.5 &  59.3 & 200.4 & 64.1 & 440.4  & 50.2 & 377.9 & 60.2  \\
& DP-Finetune & 155.5 & 14.4 & 124.0 & 16.2 & 394.2 & 25.0 & 298.1 & 26.3 & 268.3 &  61.2 & 231.6 & 66.2 & 290.3  & 60.0 & 188.7 & 59.3 \\
\cline{2-18}
& \cellcolor{gray!20}{\toolname} & \cellcolor{gray!20}{32.1} & \cellcolor{gray!20}{\textbf{69.9}} & \cellcolor{gray!20}{27.2} & \cellcolor{gray!20}{\textbf{72.3}}  & \cellcolor{gray!20}{\textbf{89.3}} & \cellcolor{gray!20}{\textbf{41.4}} & \cellcolor{gray!20}{\textbf{80.0}} & \cellcolor{gray!20}{\textbf{42.9}}  & \cellcolor{gray!20}{\textbf{27.2}} & \cellcolor{gray!20}{\textbf{85.0}} & \cellcolor{gray!20}{\textbf{24.0}} & \cellcolor{gray!20}{\textbf{90.5}} & \cellcolor{gray!20}{70.4}  & \cellcolor{gray!20}{\textbf{74.1}} & \cellcolor{gray!20}{\textbf{55.3}} & \cellcolor{gray!20}{\textbf{78.2}} \\
\hline
\multirow{6}{*}{Realistic-v6} & PE & 28.5 & 63.5 & 23.3 & \textbf{65.6} & 164.9 & 25.0 & 133.5 & 26.6 & 37.9 &  71.4 & 38.9 & 73.8 & 84.8  & 65.5 & \textbf{79.5} & 64.3 \\
& Aug-PE & \textbf{25.6} & 61.9 & \textbf{21.1} & 63.2 & 177.3 & 25.0 & 132.0 & 26.4 & 32.1 &  74.5 & 40.3 & 70.3 & \textbf{79.4} &  69.6 & 85.5  & 66.2 \\
& DP-LDM & 143.2 & 32.3 & 122.2 & 36.3 & 234.1 & \textbf{26.3} & 302.8 & 25.0 & 294.3 & 62.2 & 238.4 & 56.7 &  384.3 & 50.6 & 343.2 & 51.7 \\
& DP-LoRA & 283.2 & 12.2 & 199.3 & 13.9 & 300.8 & 25.0 & 267.4 & 27.3 & 233.2 &  64.3 & 241.8 & 61.5 &  487.2 & 52.4 & 399.3 & 56.5 \\
& DP-Finetune & 200.7 & 11.9 & 184.7 & 12.3 & 326.4 & 25.0 & 307.2 & 25.0 & 251.4 & 59.3  & 195.8 & 60.0 & 321.9  & 58.2 & 285.7 & 57.7 \\
\cline{2-18}
& \cellcolor{gray!20}{\toolname} & \cellcolor{gray!20}{35.6} & \cellcolor{gray!20}{\textbf{68.2}} & \cellcolor{gray!20}{30.8} & \cellcolor{gray!20}{\textbf{68.0}} & \cellcolor{gray!20}{\textbf{161.0}} & \cellcolor{gray!20}{25.0} & \cellcolor{gray!20}{\textbf{81.7}} & \cellcolor{gray!20}{\textbf{38.5}} & \cellcolor{gray!20}{\textbf{23.5}} & \cellcolor{gray!20}{\textbf{88.5}}  & \cellcolor{gray!20}{\textbf{19.2}} & \cellcolor{gray!20}{\textbf{92.1}} & \cellcolor{gray!20}{125.2}  & \cellcolor{gray!20}{\textbf{79.8}} & \cellcolor{gray!20}{118.6} & \cellcolor{gray!20}{\textbf{71.7}} \\
\hline
\multirow{6}{*}{Prompt2med} & PE & 24.0 & 64.8 & 61.8 & 37.5 & 84.4 & 36.3 & \textbf{39.5} & 39.1 & 44.5 & 70.1  & 39.7 & 68.4 & 41.2 & 65.3 & 40.7 & 70.2 \\
& Aug-PE & \textbf{21.3} & 65.5 & 35.0 & 67.7 & \textbf{81.5} & 36.0 & 48.2 & 36.8 & 47.3 & 72.2  & 40.0 & 72.8 &  \textbf{39.0} & 70.1 & \textbf{40.4} & 71.3 \\
& DP-LDM & 265.0 & 13.2 & 231.8 & 12.1 & 277.3 & 25.0 & 197.4 & 26.2 & 214.7 & 61.4  &  200.6 & 65.8 &  376.3 & 58.9 & 326.4 & 59.0 \\
& DP-LoRA & 374.5 & 11.4  & 265.9 & 13.2 & 282.2 & 25.0 & 243.7 & 25.0 & 299.2 & 50.0 & 235.5 & 59.9 & 388.3 & 54.7 & 321.7 & 50.5 \\
& DP-Finetune & 203.5 & 13.7 & 146.3 & 12.5 & 255.5 & 25.4 & 203.9 & 27.5 & 224.9 & 57.4 & 195.3 & 60.0 & 343.9 & 51.9 & 288.5 & 53.4 \\
\cline{2-18}
& \cellcolor{gray!20}{\toolname} & \cellcolor{gray!20}{30.2} & \cellcolor{gray!20}{\textbf{70.8}} & \cellcolor{gray!20}{\textbf{25.4}} & \cellcolor{gray!20}{\textbf{72.8}} & \cellcolor{gray!20}{108.1} & \cellcolor{gray!20}{\textbf{42.4}} & \cellcolor{gray!20}{99.1} & \cellcolor{gray!20}{\textbf{43.4}} & \cellcolor{gray!20}{\textbf{28.4}} & \cellcolor{gray!20}{\textbf{83.9}}  & \cellcolor{gray!20}{\textbf{26.0}} & \cellcolor{gray!20}{\textbf{88.5}} & \cellcolor{gray!20}{96.3} & \cellcolor{gray!20}{\textbf{78.9}}  & \cellcolor{gray!20}{89.4} & \cellcolor{gray!20}{\textbf{80.1}} \\
\bottomrule
\end{tabular}
}
\end{table*}

\begin{figure*}[htbp]
\centering
\setlength{\tabcolsep}{2pt} 
\renewcommand{\arraystretch}{1.0}

\begin{tabular}{c *{4}{c}} 
    \multicolumn{1}{c}{} & {\small\sffamily CIFAR-10} & {\small\sffamily  OCTMNIST} & {\small\sffamily  CelebA} & {\small\sffamily  Camelyon} \\

    \raisebox{0.25\height}{\rotatebox{90}{\small\sffamily  Synthetic}} 
    & \includegraphics[width=0.24\textwidth]{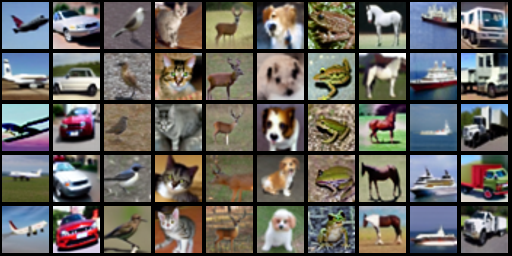}
    & \includegraphics[width=0.24\textwidth]{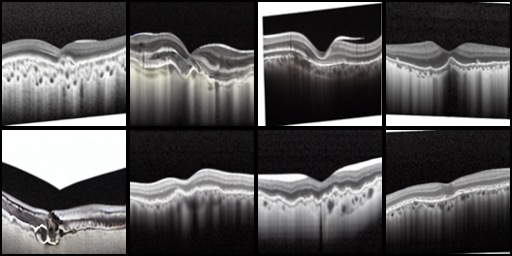}
    & \includegraphics[width=0.24\textwidth]{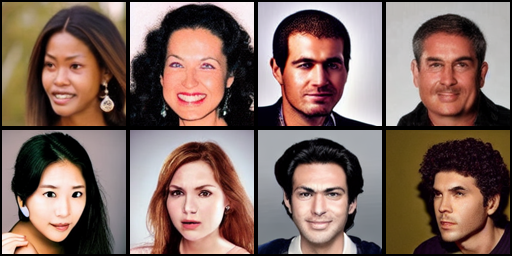}
    & \includegraphics[width=0.24\textwidth]{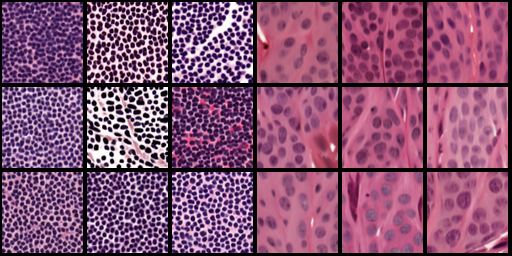} \\
    
    \raisebox{1.2\height}{\rotatebox{90}{\small\sffamily  Real}} 
    & \includegraphics[width=0.24\textwidth]{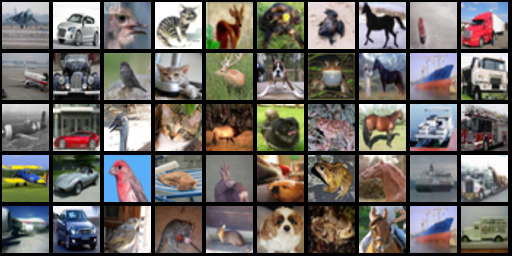}
    & \includegraphics[width=0.24\textwidth]{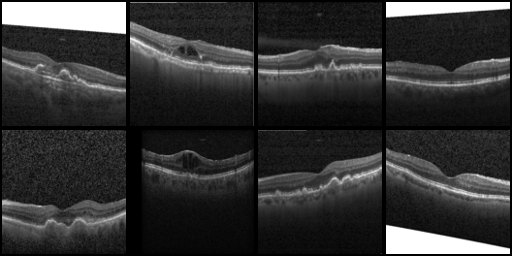}
    & \includegraphics[width=0.24\textwidth]{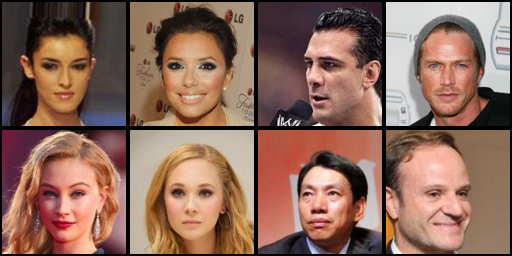}
    & \includegraphics[width=0.24\textwidth]{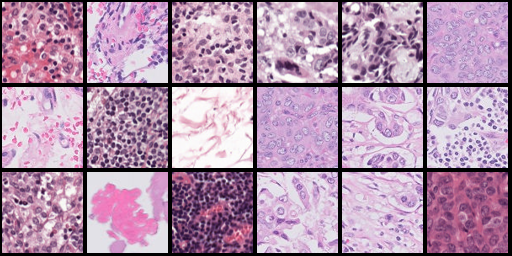} \\
\end{tabular}

\caption{Visualization examples of comparison between synthetic and real images across four datasets. The synthetic images are generated using \toolname{}, under the privacy budget $\epsilon=10.0$. }
\label{fig:synthetic_real}
\end{figure*}

\noindent \textbf{Baselines.} We elaborate on current DP image synthesis methods in~\Cref{sec:related}. This paper focuses on methods that leverage public models. Built on this constraint, we select five predominant DP image synthesis methods that use public models to aid DP image synthesis, including PE~\cite{dpsda}, Aug-PE~\cite{xiedifferentially}, DP-LDM~\cite{dpldm}, DP-LoRA~\cite{dplora}, and DP-Finetune~\cite{dp-diffusion}. 
PE and Aug-PE are fine-tuning free methods. Besides, DP-LDM, DP-LoRA, and DP-Finetune are fine-tuning-based methods.

\vspace{0.5mm}
\noindent \textbf{Implementations.} All experiments are implemented with Python 3.8 on a server with 4 NVIDIA GeForce A6000 Ada and 512GB of memory. We aim at conditional generation for these datasets (i.e., each generated image is associated with the class label). Following practical adoption in DPImageBench~\cite{gong2025dpimagebench}, we set DP parameter $\delta = 1/ (N_{\text{priv}} \times \log N_{\text{priv}})$, where $N_{\text{priv}}$ means the number of
samples in training private datasets as presented in Table~\ref{tab:datainfo}. \Cref{apsubsec:hyper} provides further details on our hyperparameter settings.

\vspace{0.5mm}
\noindent \textbf{Investigated Datasets.} We perform experiments on four image datasets {\tt CIFAR-10}~\cite{cifar10}, {\tt OCTMNIST}~\cite{yang2023medmnist}, {\tt CelebA}~\cite{celeba}, and {\tt Camelyon}~\cite{camelyon1}. The investigated datasets are prevalently used in previous DP image synthesis methods~\cite{dp-diffusion,dp-feta}.

\texttt{CIFAR-10} comprises 10 classes of natural images with the resolution of $32\times 32$. \texttt{OCTMNIST} consists of 109,309 retinal Optical Coherence Tomography (OCT) images categorized into 4 classes with the resolution of $128\times 128$~\cite{yang2023medmnist}. \texttt{CelebA} contains over 202,599 facial images of 10,177 celebrities, each annotated with 40 attributes with the resolution of $256\times 256$; following prior work~\cite{dpsda,dp-feta,gong2025dpimagebench}, we use the ``Gender'' attribute to classify images as male or female. \texttt{Camelyon} includes 455,954 histopathological image patches of human tissue with the resolution of $96\times 96$, labeled based on the presence of at least one tumor cell pixel. As shown in Table~\ref{tab:datainfo}, all datasets are split into training, validation, and test sets. It is  noticed that some prior works~\cite{dp-feta,li2023privimage} use downsampled versions of the {\tt CelebA} and {\tt Camelyon} datasets at a resolution of $32 \times 32$. In contrast, \toolname{} uses the original-resolution images, making our experiments more challenging than those in these earlier studies~\cite{gong2025dpimagebench}.

\vspace{1.0mm}
\noindent \textbf{Public Models.} We evaluate our method using four widely-used public models, including `Stable-Diffusion-v1-5\textsuperscript{\ref{sdv15}}', `Stable-Diffusion-2-1-base\footnote{\url{https://huggingface.co/Manojb/stable-diffusion-2-1-base}\label{sd21base}}', `Realistic-v6\footnote{\url{https://huggingface.co/SG161222/Realistic_Vision_V6.0_B1_noVAE}\label{realistic}}', and `Prompt2med\footnote{\url{https://huggingface.co/Nihirc/Prompt2MedImage}\label{promptmed}}'. 
Given the dominance of the stable diffusion library\footnote{\url{https://huggingface.co/stabilityai} \label{sd}}~\cite{peft,nondpfinetuning}, our selection includes two official stable diffusion base models and two specialized finetuned versions: `Realistic-v6' for photorealism and `Prompt2med' for medical imaging. \Cref{tab:publicdatainfo} presents the details of public models. 

\vspace{1.0mm}
\noindent \textbf{Evaluation Metrics.} We assess both fidelity and utility using two established metrics \cite{dpdm, dp-feta, gong2025dpimagebench}: Fréchet Inception Distance (FID) and downstream classification accuracy. Specifically, we generate 60,000 synthetic images for evaluation. Our experimental setup is grounded in the standardized evaluation framework proposed by DPImageBench~\cite{gong2025dpimagebench}.

\section{Experiment Analysis}
\label{sec:exp}
This section investigates the effectiveness of \toolname{} through answering three research questions (RQs) as follows. 
\begin{itemize}[leftmargin=*]
    \item \textbf{RQ1.} Does \toolname{} outperform baseline methods in generating high-quality images across the four diverse datasets?
    \item  \textbf{RQ2.} How does saliency-aware parameter selection in public models benefit the DP image synthesis?
    \item  \textbf{RQ3.} How do the hyperparameters introduced by \toolname{} affect the performance of synthetic images?
\end{itemize}

\subsection{Performance of Synthetic Datasets (RQ1)}
\label{sec:rq1}
This RQ evaluates the utility and fidelity of synthetic images generated by \toolname{} relative to five baselines. Using the four datasets detailed in \Cref{sec:setup}, we perform a comparative analysis under privacy budgets of $\epsilon=\{1, 10\}$. \Cref{tab:rq1} presents the FID and Acc (\%) of \toolname{} and baselines. \Cref{fig:synthetic_real} shows visualization examples of comparison between synthetic images generated by \toolname{} and real images across four studied datasets, under $\epsilon=10$.  We summarize the takeaways from this RQ as follows.

\vspace{1mm}
\noindent {\textbf{In most cases, \toolname{} generates synthetic images with higher utility and fidelity compared to baselines.}} \Cref{tab:rq1} shows that \toolname{} achieves the highest accuracy across nearly all settings; the sole exception is on {\tt OCTMNIST} under $\epsilon=1$ when using `Realistic-v6' as the public model, however, all methods do not work in this case. \toolname{} achieves an accuracy of 25.0\%, trailing the top-performing baseline (which is 26.3\%) by only a narrow margin.

Besides, in terms of FID, \toolname{} achieves lower (superior) scores on {\tt CIFAR-10} and {\tt CelebA} datasets compared to fine-tuning-free methods such as PE and Aug-PE. Specifically, since these fine-tuning-free methods preserve the pretrained weights without modification, the synthetic images strictly adhere to the strong prior distribution of the public model, resulting in high visual fidelity and lower FID scores. However, they lack deep adaptation to the label information of specific downstream tasks. When the downstream domain, such as the {\tt OCTMNIST} medical imaging dataset, exhibits a distribution shift from the natural images used in pretraining, PE-based methods struggle to capture the critical discriminative features necessary for effective classification, leading to inferior accuracy compared to \toolname{}. DPImageBench~\cite{gong2025dpimagebench} reports a similar phenomenon, noting that while the FID of synthetic images from PE is lower than that of state-of-the-art methods, its Acc is also consistently lower.

\begin{table}[!t]
\small
    \centering
    \caption{The Acc (\%) and FID of synthetic images for {\tt CIFAR-10} and {\tt CelebA}, under $\epsilon=10.0$. We only finetune the query parameter matrices of cross-attention.  The public model is `Stable-Diffusion-v1-5.'}
    \label{tab:carefully}
    \setlength{\tabcolsep}{4.2mm}{
    \resizebox{0.48\textwidth}{!}{
    \renewcommand{\arraystretch}{1.2}
    \begin{tabular}{l|cc|cc|cc}
    \toprule
    \multirow{2}{*}{\textbf{Datasets}}  & \multicolumn{2}{c|}{\textbf{DP-LoRA}} & \multicolumn{2}{c|}{\textbf{DP-LDM}}  & \multicolumn{2}{c}{\textbf{DP-Finetune}} \\
    \cline{2-7}
     & \textbf{FID} & \textbf{Acc} & \textbf{FID} & \textbf{Acc} & \textbf{FID} & \textbf{Acc} \\
    \midrule
    {\tt CIFAR-10} & 25.4 & 72.0  & 36.2 & 66.0 & 132.7 & 32.3 \\
    {\tt CelebA} & 34.0 & 85.8 & 40.6 & 80.8 & 198.0 & 61.1 \\
    \bottomrule
\end{tabular}
}}
\end{table}

\begin{figure*}[!t]
    \centering
    \setlength{\abovecaptionskip}{0pt}
    \includegraphics[width=1.0\linewidth]{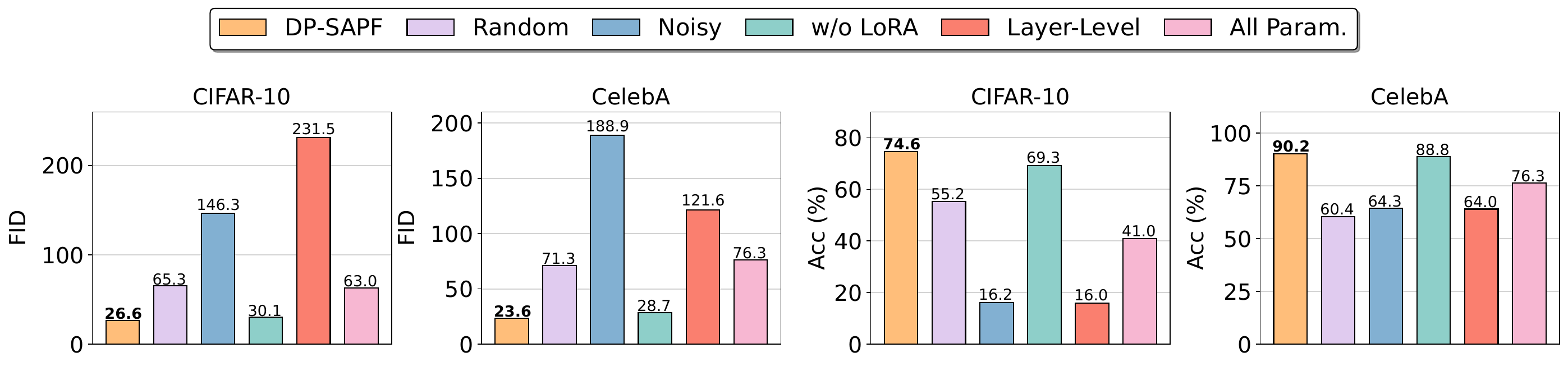}
    \vspace{1mm}
    \vspace{-4mm}
    \caption{FID and Acc (\%) of \toolname{} and five variants with $\varepsilon=10.0$. `Random' indicates no saliency-aware mechanism. `Noisy' replaces the unselected parameters with random values. `w/o LoRA' finetunes the select parameter matrices directly using DP-SGD, without incorporating LoRA. `Layer-Level' means conducting layer-wise selection. `\toolname' is our work that conducts matrix-wise saliency-aware parameter selection and finetunes the selected parameters using DP-SGD with LoRA.}
    \label{fig:rq2}
\end{figure*}

\begin{figure*}[!t]
    \centering
    \setlength{\abovecaptionskip}{0pt}
    \includegraphics[width=1.0\linewidth]{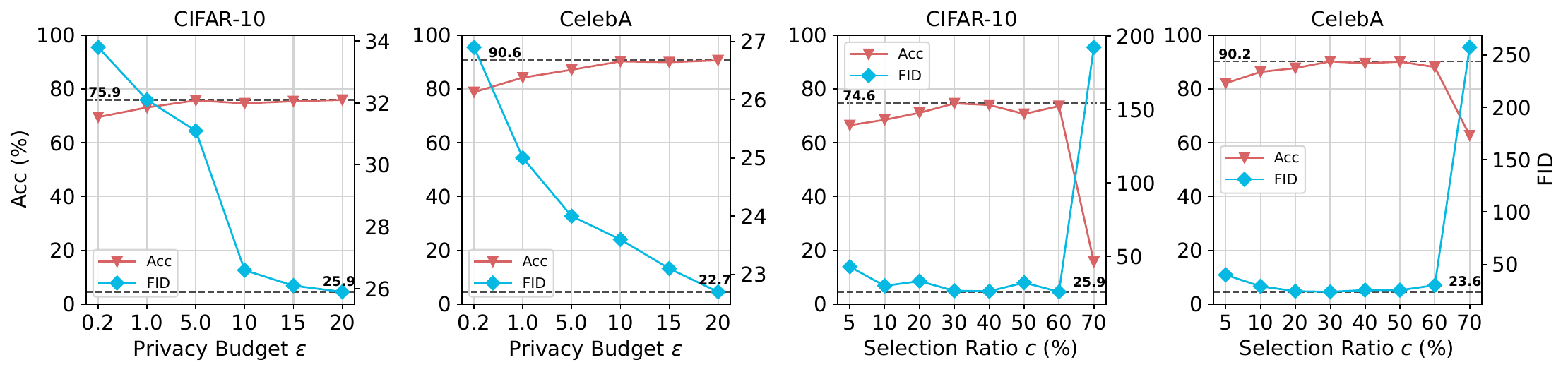}
    \vspace{1mm}
    \vspace{-4mm}
    \caption{The Acc (\%) and FID of synthetic images generated by \toolname{} for the sensitive image datasets, {\tt CIFAR-10} and {\tt CelebA}, under privacy budgets $\epsilon = \{0.2, 1.0, 5.0, 10.0, 15.0, 20.0\}$, and selection ration $c=\{5\%, 10\%, 20\%, 30\%, 40\%, 50\%, 60\%, 70\%\}$. The public model is `Stable-Diffusion-v1-5'. The dashed lines indicate the best value.}
    \label{fig:privacy_budget}
\end{figure*}

\vspace{1mm}
\noindent {\textbf{Fine-tuning all attention layers within public models often leads to training instability or model collapse.}} \Cref{tab:carefully} shows that DP-LoRA, DP-LDM, and DP-Finetune are consistently poor-performing when fine-tuning all attention layers within the public models. As discussed in~\Cref{subsec:mot}, some attention layers in public models are highly sensitive yet contribute little to fine‑tuning; updating these layers can easily destabilize the model and cause it to collapse. Through careful manual selection of parameter matrices, \Cref{tab:carefully} shows that restricting fine‑tuning to the query matrices of the cross‑attention attention layers~\cite{labelembedding} consistently enhances the quality of the generated images. However, because the public model is highly complex, manually selecting the parameter matrices for fine‑tuning is labor‑intensive and consumes additional privacy budget~\cite{koskela2023practical}.

However, \toolname{} requires no manual selection and automatically avoids training collapse. For example, under $\epsilon=10.0$ and using `Stable-Diffusion-v1-5' as the public model, the FID is reduced by 83.1\%, and the accuracy increases by 31.5\% compared to DP-LoRA averaged across four sensitive datasets.

\begin{table*}[t]
\centering
\caption{Selective fine-tuning configuration across different selection ratios $c$. 
Columns grouped by U-Net blocks (Down/Mid/Up). 
Each cell shows finetuned projections for self-attention (attn1) and cross-attention (attn2) as \texttt{attn1/attn2} 
(e.g., \texttt{v/qkv} = finetune \textbf{v} in attn1 and \textit{q, k, v} in attn2). 
``--'' indicates module not selected. The public model is `Stable-Diffusion-v1-5'. `D$i$-A$j$' = Down block $i$, attention layer $j$;  `U$i$-A$j$' = Up block $i$, attention layer $j$.  `M-A0' = Middle block, attention layer 0.} 
\label{tab:selective_finetuning}
\resizebox{\textwidth}{!}{
\begin{tabular}{lcccccccccccccccc}
\toprule
\multirow{2}{*}{\textbf{Ratio}} & 
\multicolumn{6}{c}{\textbf{Down Blocks}} & 
\multicolumn{1}{c}{\textbf{Mid Block}} & 
\multicolumn{9}{c}{\textbf{Up Blocks}} \\
\cmidrule(lr){2-7} \cmidrule(lr){8-8} \cmidrule(l){9-17}
& D0-A0 & D0-A1 & D1-A0 & D1-A1 & D2-A0 & D2-A1 & M-A0 & U1-A0 & U1-A1 & U1-A2 & U2-A0 & U2-A1 & U2-A2 & U3-A0 & U3-A1 & U3-A2 \\
\midrule
5\% & \texttt{--/v} & \texttt{--/--} & \texttt{--/v} & \texttt{--/--} & \texttt{--/v} & \texttt{--/v} & \texttt{--/--} & \texttt{--/v} & \texttt{--/v} & \texttt{--/--} & \texttt{--/--} & \texttt{--/--} & \texttt{--/--} & \texttt{--/--} & \texttt{--/--} & \texttt{--/--} \\
10\% & \texttt{v/v} & \texttt{v/--} & \texttt{--/v} & \texttt{--/v} & \texttt{--/v} & \texttt{--/v} & \texttt{--/--} & \texttt{--/v} & \texttt{--/v} & \texttt{--/v} & \texttt{--/--} & \texttt{--/--} & \texttt{--/--} & \texttt{--/--} & \texttt{--/--} & \texttt{--/--} \\
20\% & \texttt{v/v} & \texttt{v/v} & \texttt{v/v} & \texttt{--/v} & \texttt{--/v} & \texttt{--/v} & \texttt{--/v} & \texttt{--/v} & \texttt{--/v} & \texttt{--/v} & \texttt{--/v} & \texttt{--/v} & \texttt{--/v} & \texttt{--/--} & \texttt{v/--} & \texttt{v/v} \\
30\% & \texttt{v/v} & \texttt{v/v} & \texttt{v/v} & \texttt{v/v} & \texttt{v/v} & \texttt{v/v} & \texttt{--/v} & \texttt{--/v} & \texttt{--/v} & \texttt{--/v} & \texttt{--/v} & \texttt{--/v} & \texttt{v/v} & \texttt{v/v} & \texttt{v/v} & \texttt{v/v} \\
40\% & \texttt{kv/v} & \texttt{kv/v} & \texttt{v/v} & \texttt{v/v} & \texttt{v/v} & \texttt{v/v} & \texttt{--/v} & \texttt{--/v} & \texttt{--/v} & \texttt{--/v} & \texttt{v/v} & \texttt{v/v} & \texttt{v/v} & \texttt{v/v} & \texttt{v/v} & \texttt{v/qv} \\
50\% & \texttt{kqv/kqv} & \texttt{kqv/v} & \texttt{v/v} & \texttt{v/v} & \texttt{v/v} & \texttt{v/kv} & \texttt{--/v} & \texttt{--/v} & \texttt{--/v} & \texttt{v/v} & \texttt{v/v} & \texttt{v/v} & \texttt{v/v} & \texttt{v/v} & \texttt{v/v} & \texttt{v/qv} \\
60\% & \texttt{kqv/kqv} & \texttt{kqv/qv} & \texttt{kqv/v} & \texttt{kv/kv} & \texttt{v/kv} & \texttt{v/kv} & \texttt{--/v} & \texttt{v/v} & \texttt{v/v} & \texttt{v/v} & \texttt{v/v} & \texttt{v/v} & \texttt{v/v} & \texttt{v/kqv} & \texttt{v/v} & \texttt{v/qv} \\
70\% & \texttt{kqv/kqv} & \texttt{kqv/kqv} & \texttt{kqv/qv} & \texttt{kqv/kv} & \texttt{v/kv} & \texttt{v/kv} & \texttt{v/kv} & \texttt{v/v} & \texttt{v/v} & \texttt{kv/v} & \texttt{v/kqv} & \texttt{v/v} & \texttt{v/kv} & \texttt{v/kqv} & \texttt{v/kqv} & \texttt{v/qv} \\
\bottomrule
\end{tabular}
}
\end{table*}

\subsection{Strengths of Parameter Selection (RQ2)}
\label{sec:rq2}

We explore the strengths of leveraging the parameter-selection mechanism introduced by \toolname{} to select critical parameter matrices for private finetuning. We compare the performance of \toolname{} with the following five invariants of \toolname{} under $\epsilon=10.0$. In this RQ, we use the `Stable-Diffusion-v1-5' as the public model.
\begin{itemize}[leftmargin=*]
    \item `Random' denotes randomly selecting the parameter matrices without the saliency-aware mechanism introduced by \toolname, to study the importance of the saliency-aware.
    \item `Noisy' indicates that the unselected parameters are replaced with random values, allowing us to analyze how much the original generative capability of the public model contributes to the final synthetic performance.
    \item `w/o LoRA' means that after the saliency-aware mechanism, we finetune the select parameter matrices directly using DP-SGD, without incorporating LoRA.
    \item `Layer-Level' means conducting layer-wise selection, instead of the matrix-wise selection, to validate our analysis in~\Cref{subsec:mot}.
    \item `All Parameter (Param.)' indicates that the saliency‑aware selection mechanism is applied to all fine‑tunable parameter matrices, rather than being limited to the attention layers.
\end{itemize}
\Cref{fig:rq2} presents the FID and Acc (\%) of \toolname{} and five variants. We summarize the takeaways of this RQ as follows.

\begin{table*}[!t]
\renewcommand{\arraystretch}{1.1}
\setlength{\tabcolsep}{8.5pt}
\small
    \centering
    \caption{FID and Acc (\%) of \toolname{} on four sensitive datasets with $\varepsilon=\{10,\infty\}$ using two public models `Stable Diffusion-v1-5' and `Realistic-v6'. `$\epsilon = \infty$' denotes public model fine-tuning without DP protection.}
    \label{tab:no_dp}
    \resizebox{0.99\textwidth}{!}{
    \begin{tabular}{l|cc|cc|cc|cc|cc|cc|cc|cc}
    \toprule
    \multirow{3}{*}{\textbf{Method}} &  \multicolumn{8}{c|}{\textbf{Stable Diffusion-v1-5}} & \multicolumn{8}{c}{\textbf{Realistic-v6}}  \\
    \cline{2-17}
    & \multicolumn{2}{c|}{{\tt CIFAR-10}} & \multicolumn{2}{c|}{{\tt OCTMNIST}}  &  \multicolumn{2}{c|}{{\tt CelebA}} & \multicolumn{2}{c|}{{\tt Camelyon}} & \multicolumn{2}{c|}{{\tt CIFAR-10}} & \multicolumn{2}{c|}{{\tt OCTMNIST}}  &  \multicolumn{2}{c|}{{\tt CelebA}} & \multicolumn{2}{c}{{\tt Camelyon}} \\
    \Xcline{2-17}{0.5pt}
     & \centering \textbf{FID} & \textbf{Acc} & \textbf{FID} & \textbf{Acc} & \textbf{FID} & \textbf{Acc} & \textbf{FID} & \textbf{Acc} & \centering \textbf{FID} & \textbf{Acc} & \textbf{FID} & \textbf{Acc} & \textbf{FID} & \textbf{Acc} & \textbf{FID} & \textbf{Acc} \\
    \midrule
    \toolname{} $(\epsilon=10)$  & 26.6 & 74.6 & 77.9 & 46.2  & 23.6 & 90.2  & 145.9 & 79.9 &  30.8 & 68.0 & 81.7 & 38.5 & 19.2 & 92.1 & 118.6 & 71.7 \\
    \toolname{} $(\epsilon=\infty)$ & 21.4 & 76.9 & 71.8 & 50.3 & 19.2 & 93.4 & 118.0  & 81.2 & 29.8 & 70.7 &  73.3 & 42.2 & 16.0 & 95.0 & 139.3 & 72.4 \\
    \bottomrule
\end{tabular}
}
\end{table*}

\vspace{1mm}
\noindent \textbf{Matrix-wise saliency-aware selection benefits the synthetic performance.} In~\Cref{fig:rq2}, we first observe that randomly selecting the parameter matrices leads to noticeably worse synthetic images compared to saliency‑aware selection. In {\tt CIFAR‑10}, the FID and Acc are 146.3 and 16.2, which are far worse than the 26.6 and 74.6 achieved by \toolname. Random selection is unlikely to identify parameters that meaningfully influence the model’s generative behavior, causing the DP‑SGD noise to be injected into uninformative or even irrelevant parts of the model, ultimately harming both fidelity and utility. Additionally, we observe that matrix‑wise selection outperforms layer‑wise selection, which validates our analysis in~\Cref{subsec:mot} that different parameter matrices contribute unequally even within the same layer. Matrix‑wise selection enables \toolname{} to isolate and finetune the most influential matrices, reducing unnecessary noise injection into less relevant components. On {\tt CelebA}, the FID and Acc for layer‑wise selection are 121.6 and 64.0, which are greatly worse than the 23.6 and 90.2 achieved by matrix‑wise selection.

\vspace{1mm}
\noindent \textbf{LoRA fine‑tuning with DP‑SGD outperforms directly applying DP‑SGD to the selected parameter matrix set.} Previous studies~\cite{dplora,li2025lorada} show that applying LoRA finetuning with DP-SGD reduces the injected noise compared to standard DP-SGD training. Our experiments shown in~\Cref{fig:rq2} empirically validate this claim. On the {\tt CIFAR-10} and {\tt CelebA} sensitive image datasets, our method produces superior image quality compared to applying DP-SGD directly to the selected parameter matrices during fine-tuning.

\subsection{Hyperparameter Analysis (RQ3)}
\label{sec:rq3}

This RQ explores how different hyperparameter settings affect the system, considering the following two perspectives.

\begin{itemize}[leftmargin=*]
    \item \textbf{Privacy allocations.} We evaluate the performance of \toolname{} under different parameter selection ratios $c= \{5\%, 10\%, 20\%, 30\%, 40\%, 50\%, 60\%, 70\%\}$, under the privacy budget $\epsilon=10.0$. The selection ratio determines the ratio of matrices with the largest noisy gradient norms that are selected for subsequent private fine-tuning from the parameter matrices pool.
    \item \textbf{Privacy budget.} We evaluate the utility and fidelity of the synthetic images under $\epsilon \in \{0.2, 1.0, 5.0, 10.0, 15.0, 20.0\}$. To ensure a consistent privacy allocation ratio across settings, we adjust the noise scale $\sigma_s$ in the parameter‑selection mechanism so that it matches the allocation used in the $(\epsilon = 5,\ \sigma_s = 10)$ configuration reported in~\Cref{tab:ratio}.
    \item \gc{\textbf{Noise scale $\sigma_s$.} As introduced in~\Cref{subsec:privacy_analysis}, DP‑SAPF doesn’t tune privacy‑budget allocation hyperparameters. The user‑specified parameter is the noise scale $\sigma_s$  for parameter selection. The privacy allocation is automatically determined via RDP based on the dataset size, the target $\varepsilon$, and $\sigma_s$. We evaluate $\sigma_s = \{5.0,10.0,20.0,25.0\}$ under privacy budget $\epsilon=\{1.0,10.0\}$ for {\tt CelebA}. \Cref{tab:ratio} presents the DP‑cost ratios (\%) between the parameter‑selection queries and the DP‑SGD fine‑tuning stage in \toolname{}.}
\end{itemize}

\noindent All experiments in this research question are conducted using the public model Stable-Diffusion-v1-5. \Cref{fig:privacy_budget} reports the Acc (\%) and FID scores of synthetic images generated by \toolname{} under different privacy budgets and selection ratios.
\Cref{tab:selective_finetuning} lists the specific parameter matrices selected from the candidate pool for each selection ratio. %
We summarize the key findings from this RQ as follows.

\vspace{1mm}
\noindent {\textbf{The synthetic performance is relatively insensitive to the selection ratio $c$ within the range of $[20\%, 50\%]$.}} \Cref{fig:privacy_budget} shows that when the $c$ lies within the range of $[20\%, 50\%]$, both FID and Acc remain nearly unchanged, indicating that our method is largely insensitive to this hyperparameter. Consequently, tuning the selection ratio in practice is straightforward. As the increase in $c$, \toolname{} may select some sensitive parameter matrices, and finetune them, which easily collapses the public model. For example, in {\tt CelebA}, \Cref{fig:privacy_budget} shows a sharp drop in Acc and an increase in FID when the select ratio increases from $60\%$ to $70\%$. \gc{\Cref{apsec:additional_analysis} examine the sensitivity of the $c$ on more datasets and public models. }

\vspace{1mm}
\noindent {\textbf{An $\epsilon \geq 10$ provides a practical balance between synthetic performance and privacy preservation.}} In~\Cref{fig:privacy_budget}, we observe that increasing the privacy budget $\epsilon$ consistently improves synthetic image quality, as a larger $\epsilon$ corresponds to injecting less DP noise~\cite{gong2025dpimagebench,dp}. On both {\tt CIFAR‑10} and {\tt CelebA}, performance plateaus when $\epsilon \geq 10$, suggesting that $\epsilon = 10$ provides a practical balance between synthetic performance and privacy preservation. For {\tt CIFAR‑10}, increasing $\epsilon$ from 10 to 20 only reduces the FID from 26.6 to 25.5 (a marginal change of 1.1), while the Acc remains nearly unchanged.

\begin{table}[!t]
\small
    \centering
    \caption{\gc{FID and Accuracy (\%) of synthetic images under different noise scale $\sigma_s=\{5.0,10.0,20.0,25.0\}$.} }
    \label{tab:ratio_performance}
    \setlength{\tabcolsep}{2.3mm}{
    \resizebox{0.48\textwidth}{!}{
    \renewcommand{\arraystretch}{1.2}
    \begin{tabular}{l|cc|cc|cc|cc}
    \toprule
    \multirow{2}{*}{\texttt{CelebA}}  & \multicolumn{2}{c|}{$\sigma_s = 5$} & \multicolumn{2}{c|}{$\sigma_s = 10$}  & \multicolumn{2}{c|}{$\sigma_s = 20$}   & \multicolumn{2}{c}{$\sigma_s = 25$} \\
    \Xcline{2-9}{0.5pt}
    & \centering \textbf{FID} & \textbf{Acc} & \textbf{FID} & \textbf{Acc} & \textbf{FID} & \textbf{Acc} & \textbf{FID} & \textbf{Acc} \\
    \midrule
    $\epsilon = 1$ & 25.0 & 84.2 & 25.9  & 83.5&  25.8 & 84.8  & 25.8  & 83.8\\
    $\epsilon = 10$ &  23.6 & 90.2 & 23.8 & 91.5 & 23.7  & 90.1  & 23.5  & 91.1\\
    \bottomrule
\end{tabular}
}}
\end{table}

\vspace{1mm}
\noindent \textbf{Saliency-aware selection reveals structured sparsity in DP LoRA.} As present in~\Cref{tab:selective_finetuning}, cross-attention (\texttt{attn2}) and value projection ($W_v$) dominate early selection: at $c=5\%$, only the cross-attention is activated, while self-attention is frozen. $W_k$ and $W_q$ remain frozen until $c \geq 40\%$, with $W_v$ being the sole projection consistently selected. This highlights the critical role of cross-attention and value projection in DP-LoRA, as their gradients carry the strongest learning signals, especially under privacy constraints. Besides, the projection importance is layer-dependent. For example, shallow blocks (D0) require holistic tuning (\texttt{kqv/kqv} at $50\%$ and \texttt{kv/v} at 40\%), while deeper blocks (Mid and Up) benefit from sparse $W_v$-only updates. $W_q$ and $W_k$ are consistently excluded from the middle block when $c \leq 60\%$.  We also observe that at lower selection ratios, \toolname{} tends to prioritize parameter matrices proximal to the input and output layers of the model.

Notably, \toolname{} gets unstable when $c=70\%$ (as shown in \Cref{fig:privacy_budget}). This shows that the instability of DP LoRA (as presented in~\Cref{tab:moti}) is not caused by the whole middle block. Its value projection can also be finetuned stably.

\vspace{1mm}
\noindent \gc{\textbf{The synthetic performance is insensitive to the noise scale $\sigma_s$.} As shown in \Cref{tab:ratio_performance}, increasing $\sigma_s$ from $5$ to $25$ has little impact on {\tt CelebA}. Under $\epsilon=1$, the FID remains within $25.0$ to $25.9$, and the Acc stays within $83.5\%$ to $84.8\%$. Under $\epsilon=10$, the FID remains within $23.5$ to $23.8$, and the Acc stays within $90.1\%$ to $91.5\%$. The selection stage is a one-shot Gaussian query, while the final synthetic performance is mainly determined by the subsequent DP-SGD fine-tuning. As $\sigma_s$ increases from $5$ to $25$, the selection-stage cost drops from $12.92\%$ to $0.42\%$ under $\epsilon=1$, and from $0.17\%$ to $0.01\%$ under $\epsilon=10$. Hence, a larger $\sigma_s$ makes selection noisier but leaves more privacy budget for DP-SGD. Since the top-$c$ high-saliency matrices have relatively strong gradient signals, this extra selection noise mainly affects matrices near the selection boundary and does not substantially change the selected subspace. $\sigma_s=5$ is a practical default. Thus, \Cref{tab:hyperparameter} show that $\sigma_s = 5$ is used across all cases. }

\begin{table}[!t]
\small
    \centering
    \caption{RDP cost ratios (\%) of parameter-selection query /  DP-SGD in \toolname{} under different noise scale $\sigma_s$ for {\tt CelebA}, under the $\epsilon=\{1.0,5.0,10.0\}$. }
    \label{tab:ratio}
    \setlength{\tabcolsep}{1.4mm}{
    \resizebox{0.48\textwidth}{!}{
    \renewcommand{\arraystretch}{1.2}
    \begin{tabular}{l|cccc}
    \toprule
    \textbf{Privacy Budget}   & $\sigma_s = 5$ & $\sigma_s = 10$  & $\sigma_s = 20$   & $\sigma_s = 25$ \\
    \midrule
    $\epsilon = 1$ & 12.92 / 87.08 & 2.71 / 97.29 &   0.65 / 99.35 & 0.42 / 99.58\\
    $\epsilon = 5$ &  0.57 / 99.43 & 0.14 / 99.86 &  0.03 / 99.97 & 0.02 / 99.98\\
    $\epsilon = 10$ &  0.17 / 99.83 & 0.04 / 99.96 &  0.02 / 99.98 & 0.01 / 99.99\\
    \bottomrule
\end{tabular}
}}
\end{table}

\section{Discussions}
\label{sec:dis}

This section discusses (1) the impact of DP on the \toolname, (2) \gc{the transferability of \toolname,} (3) the effectiveness of fine-tuning on sensitive datasets, (4) a comparative analysis of computational overhead relative to baselines, and (5) the limitations of our methods.

\begin{table}[!t]
\renewcommand{\arraystretch}{1.1}
\setlength{\tabcolsep}{6.5pt}
\small
    \centering
    \caption{FID and Accuracy (\%) of synthetic images generated by public models, with and without fine-tuning on sensitive images under the privacy budget $\epsilon=10.0$.}
    \label{tab:no_finetuning}
    \resizebox{0.49\textwidth}{!}{
    \begin{tabular}{l|cc|cc|cc|cc}
    \toprule
    \multirow{2}{*}{\textbf{Public Model}} & \multicolumn{2}{c|}{{\tt CIFAR-10}} & \multicolumn{2}{c|}{{\tt OCTMNIST}}  &  \multicolumn{2}{c|}{{\tt CelebA}} & \multicolumn{2}{c}{{\tt Camelyon}} \\
    \Xcline{2-9}{0.5pt}
     & \centering \textbf{FID} & \textbf{Acc} & \textbf{FID} & \textbf{Acc} & \textbf{FID} & \textbf{Acc} & \textbf{FID} & \textbf{Acc} \\
    \hline
    \rowcolor[RGB]{234, 234, 234} \multicolumn{9}{c}{\textit{No Fine-tuning on Sensitive Images}} \\
    \hline
    SD-v1-5 & 38.2 & 62.6 & 239.3 & 25.0 & 110.0 & 89.6 & 406.5 & 61.1 \\
    SD-2-1-base   & 43.1 & 58.1 & 263.1 & 25.0 & 132.4 & 88.6 & 416.8 & 67.2 \\
    Realistic-v6 & 48.3 & 38.5 &  313.5 & 25.0 & 88.8 & 71.0 & 393.5 & 62.4 \\
    Prompt2med  & 47.8 & 51.2 &  128.1 & 25.0 & 139.6 & 87.7 & 420.1 & 62.5 \\
    \hline
    \rowcolor[RGB]{234, 234, 234} \multicolumn{9}{c}{\textit{Fine-tuning on Sensitive Images (\toolname)}} \\
    \hline
    SD-v1-5 & 26.6 & 74.6 & 77.9 & 46.2  & 23.6 & 90.2  & 145.9 & 79.9 \\
    SD-2-1-base   & 27.2 & 72.3 & 80.0 & 42.9 & 24.0 & 90.5 & 55.3 & 78.2  \\
    Realistic-v6 & 30.8 & 68.0 & 81.7 & 38.5 & 19.2 & 92.1 & 118.6 & 71.7 \\
    Prompt2med  & 25.4 & 72.8 & 99.1 & 43.4 & 26.0 & 88.5 & 89.4 & 80.1 \\
    \bottomrule
\end{tabular}
}
\end{table}
\vspace{-2mm}

\subsection{\toolname{} in the Non-Private Setting}
\label{subsec:non-DP}

This experiment evaluates the impact of DP on the generative performance of \toolname{}. We compare our approach against a non-DP baseline ($\epsilon=\infty$), in which models are trained using the \toolname{} without the injection of Gaussian noise.

\Cref{tab:no_dp} shows that, under a privacy budget of $\epsilon=10$, when using `Stable Diffusion-v1-5' as the public model, \toolname{} achieves average Acc reductions of only 2.3\% $(=(76.9-74.6)\times 100\%)$, 4.1\% $(=(50.3-46.2)\times 100\%)$, 3.2\% $(=(93.4-90.2)\times 100\%)$, and 1.3\% $(=(81.2-79.9)\times 100\%)$ across four sensitive datasets compared to the non-DP setting ($\epsilon=\infty$). However, we observe that for the {\tt OCTMNIST} and {\tt Camelyon} datasets, non‑DP image synthesis still performs poorly. This is primarily due to the severe distribution mismatch between the public models and these sensitive datasets. As shown in~\Cref{tab:no_finetuning}. Without fine‑tuning on the sensitive images, the synthetic images generated directly from public models exhibit poor FID and Acc. These results indicate the need for further improvements to \toolname{}. 

\subsection{Synthesizing Images without Fine-Tuning}
\label{subsec:non-finetuning}

\Cref{tab:no_finetuning} compares the synthesis quality of four public diffusion models before and after DP fine-tuning via \toolname{} on sensitive data under $\epsilon=10.0$ across four datasets. For every model–dataset pair, fine-tuning yields lower FID and higher accuracy compared to the pretrained models. \toolname{} consistently improves synthesis quality across four public models, validating its robustness to diverse pretraining priors. 

The domain bias of public models influences the absolute performance after fine-tuning. `Realistic-v6,' pretrained on high-fidelity human faces, maintains superior {\tt CelebA} synthesis (19.2 with FID) after adaptation, inheriting its strong facial prior. In contrast, despite being a medical-domain model with the best FID before fine-tuning, `Prompt2med' shows relatively modest gains on {\tt OCTMNIST} (99.1 with FID) compared to other public models, suggesting its pretraining distribution does not fully align with retinal OCT images. These results also suggest that a public model with strong downstream performance (FID or Acc) does not necessarily have greater benefits after fine-tuning. Selecting the most suitable public model remains an open question for \toolname.

\subsection{Computational Resources}

\begin{table}[!t]
\setlength{\tabcolsep}{3.0pt}
\small
\centering
\caption{GPU memory usage and runtime analysis for the {\tt CIFAR-10} using `Stable-Diffusion-v1-5' as the public model. ‘Memory’ means the GPU memory usage, and ‘Peak Memory’ means the peak GPU memory usage across all stages.}
\setlength{\tabcolsep}{3.0mm}{
\resizebox{0.47\textwidth}{!}{
\begin{tabular}{l|l|rrc}
\toprule
\textbf{Algorithm} & \textbf{Stage} & \textbf{Memory} & \textbf{Runtime} & \textbf{Peak Memory} \\
\midrule
\multirow{3}{*}{PE}     & Selecting & 0GB & 0h & \\
                        & Finetune & 0GB & 0h & 31.3GB\\
                         & Synthesis & 31.2GB & 43.5h & \\
\hline
\multirow{3}{*}{Aug-PE}   & Selecting & 0GB & 0h & \\
                        & Finetune & 0GB & 0h & 31.2GB \\
                         & Synthesis & 31.2GB & 43.5h & \\
\hline
\multirow{3}{*}{DP-LDM}    & Selecting & 0GB & 0h & \\
                            & Finetune & 33.3GB & 4.9h & 33.3GB\\
                           & Synthesis & 17.1GB & 1.9h & \\
\hline
\multirow{3}{*}{DP-LoRA}    & Selecting & 0GB & 0h & \\
                            & Finetune & 29.0GB & 4.7h & 29.0GB\\
                           & Synthesis & 17.1GB & 1.9h & \\
\hline
\multirow{3}{*}{DP-Finetune} & Selecting & 0GB & 0h & \\ 
                            & Finetune & 41.0GB & 8.0h & 41.0GB\\
                           & Synthesis & 17.1GB & 1.9h & \\
\hline
\multirow{3}{*}{\toolname} & Selecting & 4.7GB & 0.6h & \\ 
                           & Finetune & 25.9GB & 3.5h & 25.9GB\\
                           & Synthesis & 17.1GB & 1.9h & \\
\bottomrule
\end{tabular}}}
\label{tab:computationalResource}
\end{table}

This section investigates the computational resource usage of various methods. \Cref{tab:computationalResource} presents the GPU memory and runtime usage of baselines and \toolname{} for the {\tt CIFAR-10} using `Stable-Diffusion-v1-5' as the public model. 

In this table, we observe that fine-tune-free methods, such as PE~\cite{dpsda} and Aug-PE~\cite{xiedifferentially}, require substantial image synthesis and variation using a public model. When the public model-generated image is high resolution (like $256\times 256$ resolution for {\tt CelebA}), the synthesizing process is time-intensive. Thus, the training efficiency is lower than that of fine-tune-based methods. Compared with DP‑LoRA, \toolname{} updates fewer parameter matrices, which in turn reduces both training time and GPU memory consumption. Specifically, \toolname{} saves about 11.0\% $(=(29.0-25.9)/29.0 \times 100\%)$ GPU memory requirements, and $9.1\%(=(1.9+4.7-(1.9+3.5+0.6))/(1.9+4.7) \times 100\%)$ running time, compared to DP-LoRA~\cite{dplora}.

\subsection{\gc{Transferability of DP-SAPF}} 
\gc{This section evaluates whether \toolname{} can transfer to other NN structures and DP mechanisms. We evaluate \toolname{} on Diffusion-Transformer (DiT)~\cite{dit} and alternative DP mechanisms, Exponential Mechanism (EM)~\cite{dpbook} and Propose Test Release (PTR)~\cite{ptr}. 
\Cref{fig:dit} shows that \toolname{} transfers well to both alternative model architectures and alternative DP mechanisms. For DiT on {\tt CelebA}, \toolname{} improves accuracy from $61.2\%$ to $78.4\%$ and reduces FID from $305.6$ to $61.1$ compared with DP-LoRA, indicating that saliency-aware selection is not limited to U-Net-based diffusion models. For alternative DP mechanisms on {\tt CelebA}, EM and \toolname{} achieve similar accuracy of around $90\%$ and an FID of around 24.0, while outperforming the synthetic performance of PTR. }
\gc{Our choice is motivated by practical alignment with DP‑SGD. DP‑SAPF already requires per‑sample gradients for subsequent DP fine‑tuning. Leveraging the same clipped gradients and Gaussian noise enables implementation with clear sensitivity and RDP accounting, without introducing an additional DP mechanism or hyperparameters.} 

\subsection{Limitations}
\label{apsubsec:lim}

As analyzed in~\Cref{subsec:non-DP} and~\Cref{subsec:non-finetuning}, the domain bias between public models and sensitive images influences the absolute performance after fine-tuning. When the public model has a substantial domain mismatch with the sensitive data, fine‑tuning alone is insufficient to produce high‑quality synthetic images. Selecting the most suitable public model and addressing domain-mismatch challenges remain open questions for \toolname. Besides, the parameter matrix candidate pool in \toolname{} focuses on the attention layers. Indeed, the design of the matrix candidate pool plays a key role in the effectiveness of \toolname{}, and exploring how to construct it optimally remains a promising direction for future work. 

\gc{Although gradient magnitude is not a perfect saliency metric for parameter importance~\cite{DBLP:conf/iclr/AnconaCO018}, many widely used non-DP methods adopt it and achieve strong empirical performance despite these theoretical limitations~\cite{DBLP:conf/iclr/LeeAT19,DBLP:conf/iclr/WangZG20}.} 
\gc{Alternative selection metrics are possible. We do not adopt them in \toolname{} as they are typically iterative and more complex under DP. As shown in \Cref{apsubsec:alternative_saliency}, DP-SAPF achieves comparable or better performance while being simpler and more efficient. The primary contribution of \toolname{} lies in identifying training collapse when using public models under DP and proposing an effective and practical mitigation strategy. 
We leave the exploration of more suitable proxies to future work.}

\begin{figure}[!t]
    \centering
    \includegraphics[width=1.0\linewidth]{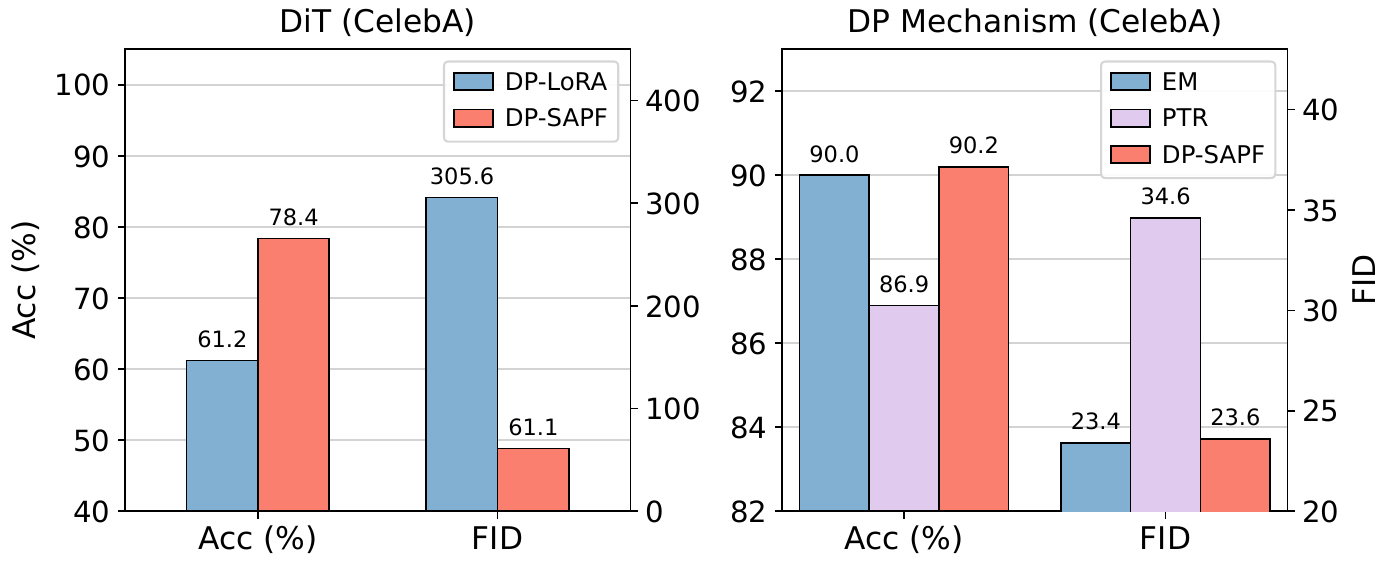}
    \caption{\gc{The synthetic performance of \toolname{} when using DiT, and alternative DP mechanisms, EM, and PTR.}}
    \label{fig:dit}
\end{figure}

\section{Related Work}
\label{sec:related}

\subsection{DP Image Synthesis}
DP image synthesis has seen various advancements to protect sensitive visual data while generating realistic synthetic images~\cite{lora,li2023privimage,dp-merf,dpdm,dp-diffusion,dp-kernel,pearl,feta-pro}. Two main types of methods include (1) using public resources like public datasets or models (like APIs) to improve the synthetic performance~\cite{xiedifferentially,li2023privimage,dp-diffusion,dpldm,dplora}; (2) while the public resource is unsuitable, another major paradigm involves training synthesizers solely on sensitive image datasets, without public resources~\cite{dp-merf,dp-feta,pearl,dp-ntk,dp-mepf,dp-kernel,dpdm}.

\vspace{1mm}
\noindent \textbf{Leveraging public dataset.} A predominant framework involves using a publicly available dataset for pretraining, followed by fine-tuning the model with sensitive images using DP-SGD~\cite{dp-mepf,dp-diffusion,dplora,dpldm,li2023privimage}. This approach leverages public data to provide a strong prior for the synthesizer. Notable methods include PDP-Diffusion~\cite{dp-diffusion}, which pretrains and finetunes using diffusion models, DP-LDM~\cite{dpldm}, using latent diffusion models, and DP-LoRA~\cite{dplora}, integrating LoRA~\cite{lora} for parameter-efficient fine-tuning.

\vspace{1mm}
\noindent \textbf{Leveraging public model.} Various methods explore using public models to reduce computational intensity, bypassing the need for the extensive computational cost of fine-tuning. The public models are either cloud-based services (e.g., Dall-E~\cite{dalle2}) or local software libraries (e.g., Stable Diffusion~\cite{labelembedding}). PE~\cite{dpsda} is a typical method that iteratively directs APIs to generate synthetic images aligning with sensitive data in the feature space. 
However, PE's performance can underperform fine-tuning-based methods when the synthetic data distribution diverges from private data. To mitigate this, SIM-PE~\cite{pe3} adjusts released public simulators (e.g., Google Font~\cite{simulator}) to fit private data under DP, offering an alternative to APIs.

Another type of DP image synthesis method using public models is fine-tuning based method~\cite{dplora,dpldm}. DP-LDM~\cite{dpldm} proposes fine-tuning the public models on the sensitive image datasets for high-quality image synthesis. DP-LoRA~\cite{dplora} then introduces the use of LoRA to finetune public models, improving training efficiency. 

\vspace{1mm}
\noindent \textbf{Selecting Public Datasets for DP Dataset Synthesis.} Previous works propose selecting public datasets from available resources to enhance synthetic data generation~\cite{li2023privimage,dpldm,yuselective}. Li et al.~\cite{li2023privimage} and Yu et al.~\cite{yuselective} propose the selection of a subset of public datasets with distributions aligned to sensitive data for pretraining. 
SIM-PE~\cite{pe3} proposes a training-free method, which selects a part of images from a large public image dataset to match the distribution of sensitive images. 

Our paper focuses on leveraging public models and proposes \toolname{} to address the problem of selecting suitable parameter matrices for private fine-tuning.

\subsection{\gc{Parameter Selection for Non-DP Synthesis}}

\gc{Parameter selection has been studied in non-DP image synthesis. A common line of work reduces the number of trainable parameters by updating only selected modules, such as attention layers or low-rank adapters, instead of fine-tuning the full model~\cite{labelembedding,peft,nondpfinetuning}. These methods are mainly designed to improve training efficiency, reduce memory cost, and preserve the generative prior of pretrained models. Other pruning-based methods use gradient-related signals to estimate parameter importance before or during training~\cite{zhang2024gradient,wang2020picking,he2023sensitivity}. 

However, parameter selection in the non-DP differs from the problem studied in DP. Non-DP methods can freely inspect the training data and choose parameters without accounting for privacy loss. In DP image synthesis, the sensitive images used for parameter selection must also be protected~\cite{dpbook}. Low-signal parameters, i.e., those with small gradients, can be easily overwhelmed by injected DP noise. Consequently, training collapse during DP fine-tuning of large public models is a DP-specific, noise-induced failure mode that is much less pronounced in non-DP training~\cite{labelembedding,peft,nondpfinetuning,zhang2024gradient,wang2020picking,he2023sensitivity}.

DP-SAPF is motivated by these differences and focuses on alleviating the DP-specific training collapse. DP-SAPF performs data-dependent matrix-level parameter selection using clipped and noised gradients, and then composes this selection privacy cost with the DP fine-tuning. In this way, DP-SAPF keeps the efficiency benefit of parameter selection while making the selection stage satisfy DP. Thus, DP-SAPF is not a straightforward adaptation of prior non-DP methods.}

\section{Conclusions}
Currently, various works propose leveraging the generative capabilities of public models to improve DP synthetic images. Existing methods commonly use LoRA to reduce the number of trainable parameters. However, we argue that applying LoRA exhaustively across all attention layers of the public model, as done in the current \textit{state‑of‑the‑art} approach, is suboptimal in a DP setting, as it introduces substantial noise accumulation and harms training stability.

This paper introduces \toolname{}, a saliency‑aware method that selects a small set of critical parameter matrices from the attention layer for LoRA fine‑tuning under DP. To identify these parameters, we first feed sensitive images into the public model, compute the corresponding gradients, and inject noise into them to satisfy DP. \toolname{} then selects the most salient parameters, those with the highest noisy gradient magnitudes, for DP fine‑tuning. Extensive experiments on four sensitive image datasets show that \toolname{} improves downstream classification accuracy of synthetic data, reduces FID, and saves computational cost, compared with finetuning approaches that update all attention parameter matrices. Even worse, fine‑tuning all attention parameter matrices is highly prone to collapse during private training. \toolname{} represents the first systematic approach to parameter selection for DP image synthesis leveraging public models, providing a foundation for future research in DP image synthesis. 

\subsection*{Ethical Considerations} 

This section examines how each stakeholder is affected during two stages: the research process (data handling) and the dissemination of results (deployment). We then outline mitigation strategies and justify the need for conducting this work.

\vspace{1mm}
\noindent \textbf{Stakeholders and Process Impact.}
\toolname{} involves three stakeholder groups.
(1) \textit{Data Subjects}: individuals represented in sensitive datasets (e.g., faces, medical scans) who rely on strong confidentiality protections during data processing.
(2) \textit{Data Owners}: institutions such as hospitals that manage sensitive images and must navigate regulatory and legal constraints during research collaborations.
(3) \textit{Researchers and Practitioners}: the ML community who depend on reproducible, methodologically reliable DP techniques.

\vspace{1mm}
\noindent \textbf{Impact of the Research.}
Releasing \toolname{} produces both benefits and risks for these groups.

\vspace{1mm}
\noindent \textit{Positive Impacts.}
(1) \textit{Supporting data sharing} (Data Owners \& Researchers): \toolname{} improves the utility of DP-generated images, enabling institutions to share information while preserving privacy, thereby expanding researchers’ access to data.
(2) \textit{Advancing transparency} (Practitioners): open-sourcing our implementation facilitates auditing and replication.

\vspace{1mm}
\noindent \textit{Negative Impacts.}
(1) \textit{Bias amplification} (Data Subjects): DP-generated data may still reflect or intensify biases present in the original datasets, potentially harming downstream individuals.
(2) \textit{Misuse risks} (Data Subjects \& Society): high-fidelity generation methods can be exploited for identity forgery, misinformation, or surveillance applications.

\vspace{1mm}
\noindent \textbf{Mitigation.}
We implement several safeguards and recommend additional measures for real-world use.

\vspace{1mm}
\noindent \textit{Methodological Mitigations.}
(1) \textit{Validation checks} (Practitioners): we apply strict verification procedures to avoid erroneous or misleading synthetic outputs.
(2) \textit{Transparency} (Data Owners \& Researchers): our released codebase enables institutions to audit privacy claims before deployment.

\vspace{1mm}
\noindent \textit{Recommended Deployment Practices.}
(1) \textit{Fairness auditing} (Data Subjects): future deployments should integrate fairness evaluations~\cite{barocas2023fairness}.
(2) \textit{Misuse prevention} (Society): controlled access and usage monitoring~\cite{tang2024modelguard}, such as restricting model checkpoints, rate-limiting, or requiring authenticated usage, can reduce the risk of malicious exploitation.

\vspace{1mm}
\noindent \textbf{Justification.}
We argue that the benefits of \toolname{} outweigh the risks given the safeguards in place. \toolname{} addresses a central challenge in privacy-preserving generative modeling: balancing utility and privacy without leveraging public models. By releasing our method and implementation, we promote transparency, facilitate community oversight, and help identify risks in DP image synthesis.

\subsection*{Open Science}

We release the replication package on the Github link.\footnote{\label{link}\url{https://github.com/2019ChenGong/DP-SAPF}} Besides, the  DOI for the artifacts is on Zenodo.\footnote{\url{https://doi.org/10.5281/zenodo.20287797}}

\bibliographystyle{ieeetr}
\bibliography{bib}

\appendix

\setcounter{section}{0}
\setcounter{equation}{0}
\renewcommand\thesection{\Alph{section}}

\section{Details of R\'{e}nyi DP in \toolname}
\label{app:supp_dp}

We use the Rényi DP (RDP) paradigm \cite{sgm} to account for the cumulative privacy costs, for fairness comparison with previous DP image synthesis methods~\cite{dp-feta,gong2025dpimagebench,dplora,dpsda}. RDP provides a rigorous accounting schedule for tracking privacy loss across multiple operations.

\begin{definition}[Sub-sampled Gaussian Mechanism (SGM~\cite{sgm})]
\label{def:sgm}
    Let $f:{D^\text{sub}} \subseteq D \to {\mathbb{R}^d}$ be a function with sensitivity ${\Delta _f} = \max_{D\simeq D'}{\left\| {f\left(D \right) - f\left({D'} \right)} \right\|_2}$.  Parameterized with a sampling rate $q \in \left( {0,1} \right]$ and noise standard deviation $\sigma>0$, the SGM $\mathcal{Q}$ is defined as,
\[
    \mathcal{Q}_{f,q,\sigma }\left( D \right) \buildrel \Delta \over = f\left( S \right) + \mathcal{N} \left( {0,{\sigma ^2}\Delta_f^2{\mathbb{I}}} \right)
    \nonumber\]
\end{definition}

\noindent where $S =$ \{${x\left| x \in D \right.}$ selected independently with probability $q$\} and $f\left( \varnothing  \right) = 0$. The privacy loss of SGM can be tracked through R\'{e}nyi DP~\cite{sgm}~(\Cref{def:rdp}). RDP can quantify the privacy loss of SGM accurately, as introduced in~\Cref{eq:rdp_gamma}. 

\begin{definition}[\textit{R\'{e}nyi DP}~\cite{sgm}]
\label{def:rdp}
The R\'{e}nyi divergence between two probability distributions $M$ and $N$ is, $D_\alpha(M \| N) = \frac{1}{\alpha - 1} \ln \mathbb{E}_{x \sim N} \left[ \left( \frac{M(x)}{N(x)} \right)^\alpha \right],$
where $\alpha > 1$ and $\alpha \in \mathbb{R}$. A randomized mechanism satisfies $(\alpha, \gamma)$-RDP if, for any neighboring datasets $D$, $D'$, and algorithm $\mathcal{Q}$, it holds that, $D_\alpha(\mathcal{Q}(D) \| \mathcal{Q}(D')) \leq \gamma$.
\end{definition}

In \toolname, prior to fine-tuning, we perform a gradient-based screening to select a subset of parameters from the public model. For each image $x_i \in D_s$, we compute the per-sample gradient $g_i$, clip it by $C_s$, and inject Gaussian noise as shown in~\Cref{subsec:sapf}. Based on the sensitivity analysis in~\Cref{subsec:sapf}, the $\ell_2$-sensitivity of the average operator is $\Delta_s = C_s / N^\ast$, where $N^\ast$ means the expected size of the sensitive datasets. Consequently, this selection phase satisfies $(\alpha, \gamma_s)$-RDP, where, $
    \gamma_s =  \frac{\alpha}{2 \sigma_s^2},$ as shown in~\Cref{the:select}.
    
Here, $\sigma_s$ is the noise multiplier specifically allocated for the parameter selection. The resulting privacy cost $\gamma_s$ remains independent of both the dataset size $n$ and the model dimensionality, providing a robust bound for the selection stage.

Then, the selected parameters are optimized using DP-SGD. For each training iteration with sampling ratio $q$, the privacy loss $\gamma_{d,i}$ is characterized by the SGM, defined in~\Cref{def:sgm}. Note that $C_s$ in the selection stage and $C$ in the training SGM correspond to the clipping bounds in their respective phases.

\begin{theorem}[]
\label{eq:rdp_gamma}
Let $p_0 = \mathcal{N}(0, C^2 \sigma^2)$ and $p_1 = \mathcal{N}(1, C^2 \sigma^2)$ denote the probability density functions of two Gaussian distributions. A single DP-SGD training step satisfies $(\alpha, \gamma_i)$-RDP for any $\gamma_i$ such that:
\begin{equation}
\gamma_i \geq D_\alpha\left( (1 - q) p_0 + q p_1 \, \| \, p_0 \right).
\end{equation}
\end{theorem}

\noindent Here the unit mean shift ($0 \rightarrow 1$) follows the normalized analysis in \cite{sgm}. \Cref{eq:rdp_gamma} facilitates the characterization of the incremental privacy cost $\gamma_i$ by evaluating the R\'{e}nyi divergence across a sub-sampled mixture distribution. The mixture distribution arises because in SGM the differing example is included in the sub-batch with probability $q$. leading to a mixture of the two output distributions. To report these guarantees within the standard $(\varepsilon, \delta)$-DP paradigm, we utilize the following analytical conversion.

\begin{theorem}[\textit{From $(\alpha, \gamma)$-RDP to $(\varepsilon, \delta)$-DP}~\cite{rdp}]
\label{eq:rdp_to_dp}
A mechanism $\mathcal{A}$ satisfying $(\alpha, \gamma)$-RDP also satisfies $(\varepsilon, \delta)$-DP for any $0 < \delta < 1$, where, $\varepsilon = \gamma + \frac{\ln (1/\delta)}{\alpha - 1}$.
\end{theorem}

By adjusting the noise variance $\sigma^2$, the privacy cost $\varepsilon = \gamma + \frac{\ln (1/\delta)}{\alpha - 1}$ can meet the target privacy budget $\varepsilon$. In practice, $\epsilon$ is obtained by optimizing over $\alpha$, typically evaluated over a predefined grid (e.g., $\alpha \in [1.01, 2, 4, ..., 256]$).

\vspace{1mm}
\noindent \textbf{Privacy Composition.} 
The sequential execution of multiple DP operations requires a unified approach to determine the global privacy guarantee. An advantage of the RDP framework is its linear compositionality. Formally, if a sequence of $k$ mechanisms $\{\mathcal{Q}_1, \mathcal{Q}_2, \ldots, \mathcal{Q}_k\}$ is applied, where each $\mathcal{Q}_i$ independently satisfies $(\alpha, \gamma_i)$-RDP, the resulting composite mechanism is guaranteed to satisfy $(\alpha, \gamma)$-RDP with $\gamma = \sum_{i=1}^k \gamma_i.$

This additive property streamlines privacy accounting across the iterative updates of DP-SGD and the multi-stage pipeline of \toolname{}. The total privacy cost accumulates as, $\gamma_{\text{total}} = \gamma_s + \sum_{i=1}^{t_d} \gamma_{d,i}$.
Once the cumulative RDP cost $(\alpha, \gamma_{\text{total}})$ is determined, it is mapped back to the standard $(\varepsilon, \delta)$-DP through \Cref{eq:rdp_to_dp} by optimizing over the order $\alpha$.

\section{Missing Proof}
\label{apsec:proof}

\textbf{Proof of~\Cref{the:select}.} \textit{The averaged clipped gradient query $[S_1(D_s),...,S_K(D_s)]$ has global $\ell_2$ sensitivity $\Delta_s = C_s / N^\ast.$ For any Rényi order $\alpha > 1$, adding Gaussian noise $\mathcal{N}\!\left(0, \sigma_s^2 \Delta_s^2 \mathbb{I}\right)$ to $S_k(D_s)$ for the weight matrices $\{\mathbf{W}_1, \dots, \mathbf{W}_K\}$ ensures that the resulting mechanism satisfies $(\alpha, \gamma_s)$-RDP for $\gamma_s = \frac{\alpha}{2\sigma_s^2}$.}

\vspace{1mm}
\noindent \textit{Proof.} Consider two neighboring sensitive datasets under the add/remove adjacency:
$D_s=\{x_1,\dots,x_{N}\}$ and $D_s'=\{x_1,\dots,x_{N-1}\}$.
For each sample $x_i$, let $
g_i=[\mathrm{vec}(\nabla_{\textbf{W}_1}\mathcal{L}_i),\dots,
      \mathrm{vec}(\nabla_{\textbf{W}_K}\mathcal{L}_i)]$ denote the concatenated gradient across all $K$ candidate matrices,
and let $\bar g_i = \Gamma_{C_s}(g_i)$ be the jointly clipped gradient. By the definition of the joint clipping operator, $\|\bar g_i\|_2 \le C_s.$

\begin{table*}[!t]
    \centering
    \caption{Hyperparameter settings of \toolname. For different public models, we use the same hyperparameter settings.}
    \label{tab:hyperparameter}
    \setlength{\tabcolsep}{14.5pt}
    \renewcommand{\arraystretch}{1.1}
    \resizebox{1.0\textwidth}{!}{
    \begin{tabular}{l|cccc|cccc}
    \toprule
    \multirow{2}{*}{\textbf{hyperparameter}}& \multicolumn{4}{c|}{$\epsilon=1.0$} & \multicolumn{4}{c}{$\epsilon=10.0$}\\
    \Xcline{2-9}{0.5pt}
    & {\tt CIFAR-10} & {\tt OCTMNIST} & {\tt CelebA} & {\tt Camelyon} & {\tt CIFAR-10} & {\tt OCTMNIST} & {\tt CelebA} & {\tt Camelyon}\\
    \midrule
    Noise scale $\sigma_d$ & 21.2 & 15.6 & 4.2 & 2.4 & 1.7 & 1.5 & 0.8 & 0.7 \\
    Noise scale $\sigma_s$ & 5.0 & 5.0 & 5.0 & 5.0 & 5.0 & 5.0 & 5.0 &  5.0\\
    Selection grad. norm $C_s$ & 1.0 & 1.0 & 1.0 & 1.0 & 1.0 & 1.0 & 1.0 & 1.0 \\
    Fine-tuning epoch & 91 & 74 & 28 & 15 & 91 & 74 & 28 & 15 \\
    Fine-tuning iterations $t_d$ & 1000 & 1000 & 1000 & 1000 & 1000 & 1000 & 1000 & 1000 \\
    Select ratio $c$ & 30\% & 30\% & 30\% & 30\% & 30\% & 30\% & 30\% & 30\% \\
    Fine-tuning learning rate $\lambda$ & $5e^{-4}$ & $5e^{-4}$ & $5e^{-4}$ & $5e^{-4}$ & $5e^{-4}$ & $5e^{-4}$ & $5e^{-4}$ & $5e^{-4}$\\
    Fine-tuning Batch size & 4096 & 4096 & 4096 & 4096 & 4096 & 4096 & 4096 & 4096 \\
    Fine-tuning grad. norm $C$ & 1.0 & 1.0 & 1.0 & 1.0 & 1.0 & 1.0 & 1.0 & 1.0 \\
    Fine-tuning sample rate $q$ &$9.1e^{-2}$ & $7.4e^{-2}$ & $2.8e^{-2}$ & $1.5e^{-2}$ & $9.1e^{-2}$ & $7.4e^{-2}$ & $2.8e^{-2}$ & $1.5e^{-2}$ \\
    LoRA rank $r$ & 4 & 4  & 4 & 4 & 4 & 4 & 4 &  4 \\
    \bottomrule
\end{tabular}
}
\end{table*}

\begin{table*}[!t]
    \centering
    \caption{Prompts used for class-conditional generation across four datasets.}
    \label{tab:prompts}
    \setlength{\tabcolsep}{23.0pt}
    \resizebox{1.0\textwidth}{!}{
    \renewcommand{\arraystretch}{1.2}
    \begin{tabular}{l|clp{9.5cm}}
    \toprule
    \textbf{Dataset} & \textbf{The Number of Category} & \textbf{Prompt} \\
    \midrule
    {\tt CIFAR-10} & 10 & 
    \makecell[l]{
        ``An image of an airplane'', ``An image of an automobile'', ``An image of a bird'',\\
        ``An image of a cat'', ``An image of a deer'', ``An image of a dog'',\\
        ``An image of a frog'', ``An image of a horse'', ``An image of a ship'',\\
        ``An image of a truck.''
    } \\
    \midrule
    {\tt CelebA} & 2  & 
    ``An image of a female face'', ``An image of a male face'' \\
    \midrule
    {\tt Camelyon} & 2  & 
    ``A normal lymph node image'', ``A lymph node histopathology image'' \\
    \midrule
    {\tt OCTMNIST} & 4  & 
    \makecell[l]{
        ``An optical coherence tomography (OCT) image for retinal disease 1'',\\
        ``An optical coherence tomography (OCT) image for retinal disease 2'',\\
        ``An optical coherence tomography (OCT) image for retinal disease 3'',\\
        ``An optical coherence tomography (OCT) image for retinal disease 4''
    } \\
    \bottomrule
    \end{tabular}}
\end{table*}

Let $F(D_s)$ denote the vector-valued query that stacks the averaged clipped gradients for all matrices:
$$
F(D_s)
=
\left[S_1(D_s),\dots,S_K(D_s)\right]
=
\frac{1}{N^\ast}\sum_{i=1}^{N} \bar g_i.
$$
Similarly,
$
F(D_s')=\frac{1}{N^\ast}\sum_{i=1}^{N-1} \bar g_i.$
Since the two datasets differ in only one sample $x_N$, their results satisfy,
$
F(D_s)-F(D_s')
=
\frac{1}{N^\ast}\bar g_N.
$
Taking the $\ell_2$ norm and using $\|\bar g_N\|_2 \le C_s$ gives,
$$
\|F(D_s)-F(D_s')\|_2
=
\frac{1}{N^\ast}\|\bar g_N\|_2
\le
\frac{C_s}{N^\ast}.
$$
Thus the global $\ell_2$ sensitivity is $\Delta_s=\frac{C_s}{N^\ast}.$ We refer to implementation in DPImageBench~\cite{gong2025dpimagebench}, we adopt the approximation $N^\ast \approx N$ and ignore the additional privacy cost associated with estimating the dataset size.

By the Rényi DP guarantee of the Gaussian mechanism, for any Rényi order $\alpha>1$,
$\tilde{F}(D_s)$ satisfies $(\alpha,\gamma_s)$-RDP with~\cite{rdp},
\begin{equation*}
    \gamma_s = \frac{\alpha \Delta_s^2}{2 (\sigma_s C_s / N^\ast)^2} = \frac{\alpha (C_s/N^\ast)^2}{2 (\sigma_s C_s / N^\ast)^2} = \frac{\alpha}{2 \sigma_s^2}.
\end{equation*}
All averaged gradients $\{S_k(D_s)\}_{k=1}^K$ are jointly protected because their concatenation $F(D_s)= [S_1(D_s),\dots,S_K(D_s)]$ has global sensitivity $\Delta_s$, and we add independent Gaussian noise of variance $\sigma_s^2 \Delta_s^2 \mathbb{I}$ to each block. This is equivalent to applying a Gaussian mechanism to the full vector $F(D_s)$, and therefore the released set $\{S_k(D_s)\}_{k=1}^K$ satisfies the same $(\alpha,\gamma_s)$-RDP.

\vspace{1mm}
\noindent \textbf{Proof of~\Cref{the:dplora}.} \textit{Given a fixed clipping threshold $C$, noise multiplier $\sigma_d$, sampling rate $q$, and number of iterations $t_d$, the $(\epsilon, \delta)$-DP guarantee of DP-SGD remains identical regardless of whether the training is performed on the full parameter set $\Theta$ or a low-rank subspace $\Theta_{\text{LoRA}}$.}

\vspace{1mm}
\noindent \textit{Proof.} For the objective $\mathcal{L}(\theta; x)$, let $g(\theta; x) = \nabla_{\theta} \mathcal{L}(\theta; x)$ denote the gradient. The clipping operator $\Gamma_C$ is defined as,
$$\Gamma_C(g) = g \left/ \max\left(1, \frac{\|g\|_2}{C}\right)\right.,$$
Let $\mathcal{M}_{\text{full}}$ be the DP-SGD over the full parameter space $\Theta \in \mathbb{R}^{m\times h}$ and $\mathcal{M}_{\text{lora}}$ be the mechanism over the LoRA subspace $\Theta_{\text{LoRA}} \in \mathbb{R}^{r\times (m+h)}$. The $m$ and $h$ are the sizes of the weight matrix, and $r \ll \min(m, h)$ is the rank of the decomposition. For identical hyperparameters $\{C, \sigma_d, q, t_d\}$, it holds that,
$$\text{PrivacyCost}(\mathcal{M}_{\text{full}}) = \text{PrivacyCost}(\mathcal{M}_{\text{lora}}).$$
Considering a single iteration, the aggregated clipped gradient for the sample $x_i$ in a sub-batch $D_s^{\text{sub}} = \{x_i\}_{i=1}^{B}$ is, $\sum_{i \in B} \Gamma_C(g_i).$
The $g_i$ means the gradient for $x_i$. The $\ell_2$-sensitivity of $\tilde{G}_t$ for neighbor datasets $D, D'$ (as introduced in~\Cref{sub:dp}) is,
$$\Delta_2 = \max_{D_s, D_s'} \left\| \sum_{i \in D_s} \Gamma_C(g_i) - \sum_{j \in D_s'} \Gamma_C(g_j) \right\|_2.$$
By the Triangle Inequality and the definition of $\Gamma_C$,
$$\Delta_2 = \max_{x \in D_s \Delta D_s'} \|\Gamma_C(g(\theta; x))\|_2 \leq C.$$
Crucially, the bound $C$ is a scalar invariant to the dimensionality of the parameter vector space $d$, i.e., $\forall d \in \{m\times h, r \times (m+h)\}: \|\Gamma_C(g)\|_2 \leq C$. Consequently, the RDP budget $\gamma_{d,i} = \frac{\alpha \Delta_2^2}{2\sigma_d^2 C^2} = \frac{\alpha}{2\sigma_d^2}$ remains independent of the number of trainable parameters. For simplicity, we omit privacy amplification here.
Thus, the dimensionality reduction inherent in LoRA does not compromise the DP. LoRA and full-parameter DP-SGD satisfy the same $(\epsilon, \delta)$-DP under equivalent noise. 

\begin{figure*}[!t]
    \centering
    \setlength{\abovecaptionskip}{0pt}
    \includegraphics[width=0.99\linewidth]{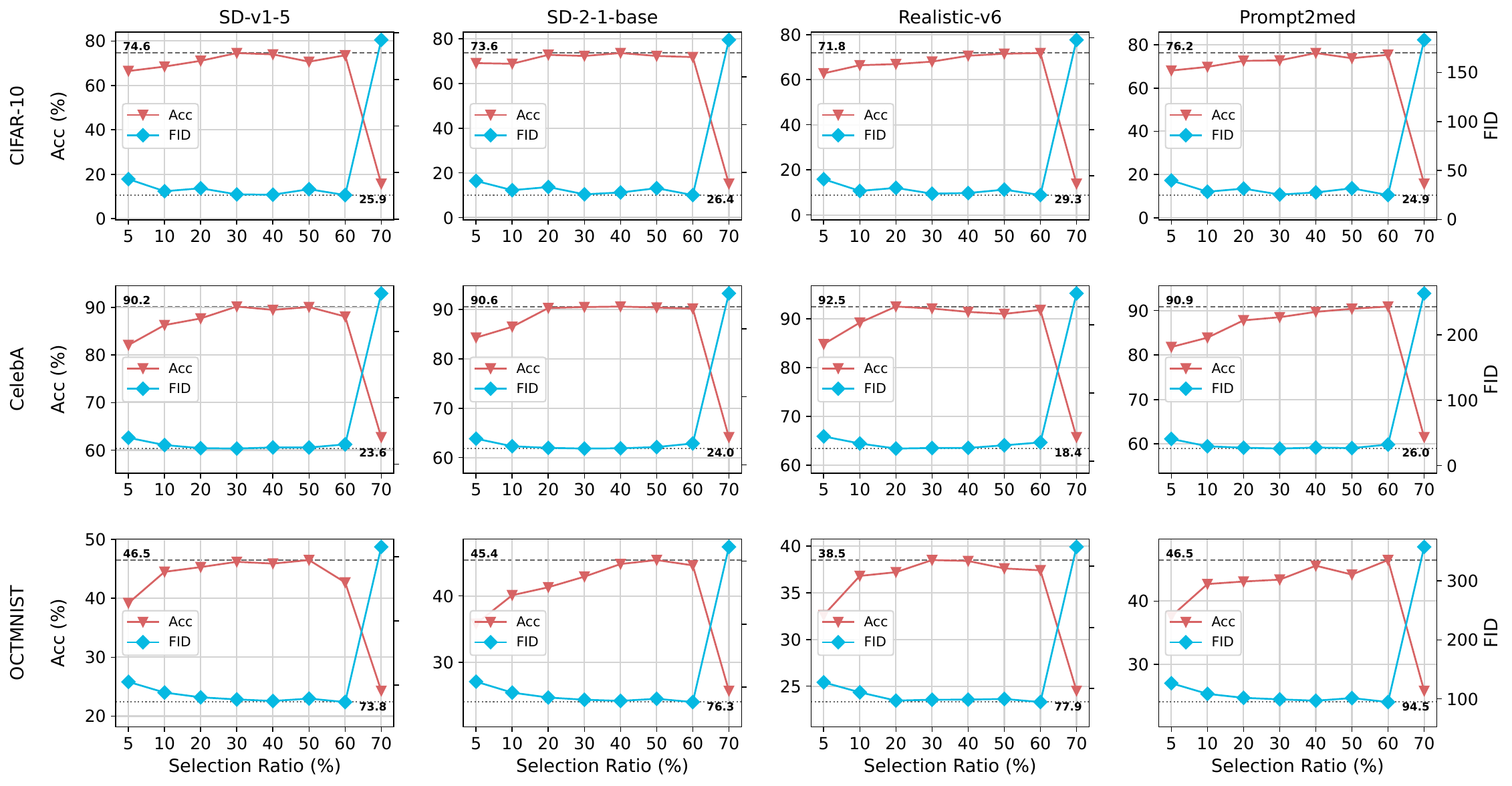}
    \caption{\gc{The Acc (\%) and FID of synthetic images generated by \toolname{} for the sensitive image datasets. }}
    \label{fig:ratio_combine}
\end{figure*}

\section{Implementation Details}
\label{apsec:id}

\subsection{Hyperparameter Settings of \toolname}
\label{apsubsec:hyper}

This section introduces the hyperparameter settings in \toolname. As presented in Table~\ref{tab:hyperparameter},  most parameters remain fixed across all datasets and privacy budgets. Only three parameters vary by dataset. The sample rate $q$ is dataset-dependent and derived from $q = \text{batch size} / \text{dataset size}$. Following previous works~\cite{dpldm,dplora}, we ignore the privacy cost of querying the dataset size. The noise scale of DPSGD $\sigma_d$ is computed via the DP composition theorem to satisfy the given privacy budget. Fine-tuning epoch naturally differs across datasets due to their varying training set sizes, while the number of fine-tuning steps  remains strictly unified. It is noticed that for different public models, we use the same hyperparameter settings.

We did not tune the shared hyperparameters between DP-LDM, DP-LoRA, DP-Finetune, and \toolname{} to avoid consuming additional privacy budget~\cite{koskela2023practical} and instead used the default settings presented in the previous work~\cite{gong2025dpimagebench}.

\subsection{Text Prompt}
\label{apsubsec:text}

Since we fine‑tune public text-to-image generative models (e.g., Stable Diffusion~\cite{labelembedding}), text prompts are conditional inputs during DP finetuning. For each dataset, we convert class labels into natural-language descriptions to form \textit{(prompt, image)} pairs for fine-tuning, as shown in~\Cref{tab:prompts}. The prompts for {\tt CIFAR-10}, {\tt CelebA}, and {\tt Camelyon} follow the mappings in DPImageBench~\cite{gong2025dpimagebench}. For the {\tt OCTMNIST} dataset, we adopt neutral numeric labels (type 1–4) rather than clinical abbreviations because we find that pretrained text encoders lack reliable semantic 
grounding for these rare ophthalmic terms.

\begin{table}[!t]
    \centering
    \caption{\gc{Performance of synthetic images using alternative saliency metrics under privacy budget $\epsilon=10.0$. We use `Stable-Diffusion-v1-5' as the public model.}}
    \label{tab:alternative-saliency}
    \setlength{\tabcolsep}{1.8mm}{
    \resizebox{0.48\textwidth}{!}{
    \renewcommand{\arraystretch}{1.2}
    \begin{tabular}{l|ccc|ccc|ccc|ccc}
    \toprule
    \multirow{2}{*}{\textbf{Datasets}}  & \multicolumn{3}{c|}{\textbf{SNIP}} & \multicolumn{3}{c|}{\textbf{FORCE}}  & \multicolumn{3}{c|}{\textbf{SynFlow}}   & \multicolumn{3}{c}{\textbf{DP-SAPF}} \\
    \Xcline{2-13}{0.5pt}
    & \centering \textbf{FID} & \textbf{Acc} & \textbf{Time} & \textbf{FID} & \textbf{Acc} & \textbf{Time} & \textbf{FID}  & \textbf{Acc} & \textbf{Time} & \textbf{FID} & \textbf{Acc} & \textbf{Time} \\
    \midrule
    {\tt CIFAR-10} &  27.1 & 74.2 & 1.2h  & 29.2 &  73.8 & 6.0h & 26.3 & 73.8 & 1.8h & 26.6 & 74.6 & 0.6h \\
    {\tt CelebA} & 24.0  & 89.8 & 1.5h  & 27.1 & 89.0  & 7.5h  &  27.5 & 89.5 & 2.2h  & 23.6 & 90.2 & 0.8h \\
    \bottomrule
\end{tabular}
}}
\end{table}

\section{\gc{Additional Experimental Analysis}}
\label{apsec:additional_analysis}

\subsection{\gc{Alternative Saliency Metric}}
\label{apsubsec:alternative_saliency}

\gc{We conduct experiments using three alternative saliency metrics, SNIP~\cite{DBLP:conf/iclr/LeeAT19}, FORCE~\cite{force}, and SynFlow~\cite{synflow}. As shown in \Cref{tab:alternative-saliency}, these metrics achieve similar FID and accuracy, suggesting that the benefit mainly comes from selecting a high-saliency parameter subspace. However, \toolname{} is consistently faster. It takes only $0.6$h/$0.8$h on {\tt CIFAR-10}/{\tt CelebA}, compared with $1.2$h/$1.5$h for SNIP, $1.8$h/$2.2$h for SynFlow, and $6.0$h/$7.5$h for FORCE. This speedup comes from the scoring procedure: \toolname{} directly ranks matrices using the clipped and noised gradients already computed for the private selection query, whereas SNIP, FORCE, and SynFlow require additional metric-specific saliency computation before DP fine-tuning. Therefore, alternative saliency metrics are feasible, but our gradient-norm criterion is simpler and more time-efficient while preserving comparable or better synthetic performance. We emphasize that the main contribution of this paper is identifying a training collapse when using public models under DP and proposing an effective, straightforward mitigation method.}

\subsection{\gc{Analysis for Parameter Selection Ratio}}
\label{apsubsec:noise_sacle}

\gc{This section evaluates the sensitivity of parameter selection ratio $c$ for the datasets {\tt CIFAR-10}, {\tt CelebA}, and {\tt OCTMNIST}, using public models `Stable-Diffusion-v1-5', `Stable-Diffusion-2-1-base', `Realistic-v6' and `Prompt2med'. 

\Cref{fig:ratio_combine} shows that when $c$ lies within the range of $[20\%, 50\%]$, both FID and accuracy remain largely stable across most cases, indicating that our method is relatively insensitive to this hyperparameter.
Although for {\tt OCTMNIST}, when using `Stable Diffusion 2.1-base' as the public model, the accuracy increases from 41.9\% to 45.4\% (a modest gain of 3.5\%), both FID and accuracy remain stable in the other cases.
Consequently, tuning the selection ratio in practice is simple, and $30\%$ or $40\%$ both serves as a suitable default choice when applying \toolname{} to new sensitive datasets and public models. This paper does not tune this hyperparameter and instead directly adopt $30\%$ as the default setting. }

\end{document}